\documentclass[aps,prx,10pt,twocolumn,showpacs,superscriptaddress,nofootinbib,longbibliography]{revtex4-2}
\usepackage{amsmath,amssymb,amsthm}
\usepackage{mathrsfs}
\usepackage[margin=1in]{geometry}
\usepackage{booktabs}
\usepackage{tabularx}
\usepackage{graphicx}
\usepackage{url}
\usepackage{wasysym}
\usepackage[justification=raggedright, singlelinecheck=false]{caption}

\newcolumntype{Y}{>{\raggedright\arraybackslash}X}

\usepackage[dvipsnames]{xcolor}
\usepackage{quantikz}
\usepackage{pgfplots}
\pgfplotsset{compat=1.18}
\usepackage{subcaption}
\usepackage[pagebackref=true]{hyperref}
\makeatletter
\let\BR@NAT@make@cite@list\NAT@make@cite@list
\def\NAT@make@cite@list{\BR@NAT@make@cite@list\Hy@backout{\@citeb}}
\makeatother
\usepackage[capitalise,nameinlink,noabbrev]{cleveref}
\crefname{lemma}{lemma}{lemmas}
\crefname{proposition}{proposition}{propositions}
\crefname{definition}{definition}{definitions}
\crefname{theorem}{theorem}{theorems}
\crefname{example}{example}{examples}
\crefname{section}{section}{sections}
\crefname{appendix}{appendix}{appendices}
\crefname{figure}{figure}{figures}
\crefname{equation}{eq.}{eqs.}
\crefname{table}{table}{tables}

\makeatletter
\backrefparscanfalse
\let\BR@pending\@empty
\let\BR@rtx@BibitemShut\BibitemShut
\def\BibitemShut#1{\let\BR@pending\backrefprint\BR@rtx@BibitemShut{#1}}
\let\BR@rtx@bibitem@fin\bibitem@fin
\def\bibitem@fin{\BR@rtx@bibitem@fin\BR@pending\let\BR@pending\@empty}
\makeatother
\renewcommand*{\backref}[1]{}
\renewcommand*{\backrefalt}[4]{%
  \ifcase #1 \or (page~#2.)\else (pages~#2.)\fi}
\usepackage{xcolor}
\usepackage{bbm}

\usetikzlibrary{arrows.meta,decorations.pathmorphing}

\definecolor{cbError}{HTML}{D55E00}
\definecolor{cboscar}{HTML}{0072B2}  
\definecolor{cbFix}{HTML}{009E73}   
\definecolor{cbWarn}{HTML}{E69F00} 

\theoremstyle{plain}
\newtheorem{theorem}{Theorem}[section]
\newtheorem{lemma}{Lemma}[section]
\newtheorem{proposition}{Proposition}[section]
\newtheorem{corollary}{Corollary}[section]

\theoremstyle{definition}
\newtheorem{definition}{Definition}[section]

\theoremstyle{remark}
\newtheorem{remark}{Remark}[section]
\theoremstyle{definition}
\theoremstyle{remark}

\theoremstyle{definition}
\newtheorem{question}{Question}
\crefname{question}{question}{questions}

\definecolor{dgYellow}{HTML}{4B2A7B}  
\definecolor{dgCyan2}{HTML}{1e98cd} 
\definecolor{dgCyan}{HTML}{2E6B3A}   
\definecolor{dgPurple}{HTML}{4B2A7B}  
\definecolor{dgViolet}{HTML}{9B1FC9} 
\definecolor{dgGreen}{HTML}{2E6B3A}  
\definecolor{dgOrange}{HTML}{EDA03A}  

\definecolor{dgBrown}{HTML}{8A5A30}

\newcommand{\Ad}{\operatorname{Ad}}
\newcommand{\ad}{\operatorname{ad}}

\newcommand{\Leb}{\operatorname{Leb}}

\newcommand{\Lie}{\mathrm{Lie}}
\newcommand{\Span}{\operatorname{span}}
\newcommand{\tr}{\operatorname{Tr}}

\newcommand{\Ucal}{\mathcal U}
\newcommand{\Comm}{\mathcal C(A,B)}
\newcommand{\Grep}{\widehat{\mathcal{G}}}

\renewcommand{\ket}[1]{\left | #1  \right \rangle}
\renewcommand{\bra}[1]{\left \langle #1 \right|}

\newcommand{\I}{\mathbbm{I}}
\newcommand{\C}{\mathbb{C}}
\newcommand{\R}{\mathbb{R}}

\newcommand{\D}{\mathcal{D}}
\newcommand{\g}{\mathfrak{g}}
\newcommand{\h}{\mathfrak{h}}
\newcommand{\tf}{\mathfrak{t}}
\DeclareMathOperator*{\esssup}{ess\,sup}

\newcommand{\fibvol}{\mathsf{V^F}}
\newcommand{\Vol}{\mathsf{V}}
\newcommand{\Hd}{\mathsf{H}}
\newcommand{\Lprod}[1]{U_{<#1}}
\newcommand{\Rprod}[1]{U_{>#1}}
\newcommand{\vol}{\operatorname{\mathsf{vol}}}
\newcommand{\timeB}{t}

\begin{document}

\title{Generating random unitaries by products of conjugated Hamiltonian evolutions}

\author{Oscar Scholin \href{https://orcid.org/0009-0005-3537-4736}{\includegraphics[scale=0.04]{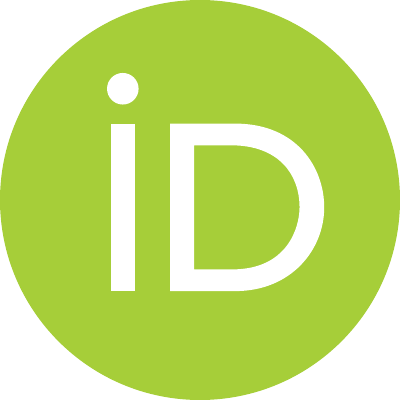}}}
\thanks{Corresponding authors: oscar.scholin@physics.ox.ac.uk, sathya.subramanian@cs.ox.ac.uk}
\affiliation{Department of Atomic and Laser Physics, University of Oxford, United Kingdom}

\author{Apollonas S. Matsoukas-Roubeas \href{https://orcid.org/0000-0001-5517-0224}{\includegraphics[scale=0.04]{orcidid.pdf}}}
\affiliation{Department of Applied Mathematics and Theoretical Physics, University of Cambridge, United Kingdom}

\author{Sathyawageeswar Subramanian \href{https://orcid.org/0000-0001-8716-3424}{\includegraphics[scale=0.04]{orcidid.pdf}}}
\thanks{Corresponding authors: oscar.scholin@physics.ox.ac.uk, sathya.subramanian@cs.ox.ac.uk}
\affiliation{Department of Computer Science, University of Oxford, United Kingdom}

\date{September 30, 2026}

\begin{abstract}
Random unitary operations are a fundamental resource for quantum information processing, yet provably generating them with experimentally available controls remains challenging. 
Existing approaches often require calculating higher-order operator expectation values, relying on statistical assumptions about the experimental Hamiltonian's spectrum to make these calculations tractable.
We show how to implement random unitaries by a sequence of evolutions generated by a fixed Hamiltonian conjugated by a global control operator. 
We give a geometric argument that establishes exponential convergence to the uniform random (Haar) measure on the accessible unitary group under uniform parameter sampling, beyond a finite sequence-length threshold that we bound.
We show this holds when the connected Lie subgroup generated by the conjugated Hamiltonians is closed. Native dynamics and global control thus provide a constructive route to uniform randomness, with applications to many experimental platforms including Rydberg atoms, bosons and fermions in optical lattices, transmons, and trapped ions.
\end{abstract}

\maketitle

\tableofcontents

\section{Introduction}\label{sec:intro}
\begin{figure*}[!ht]
    \includegraphics[width=\linewidth]{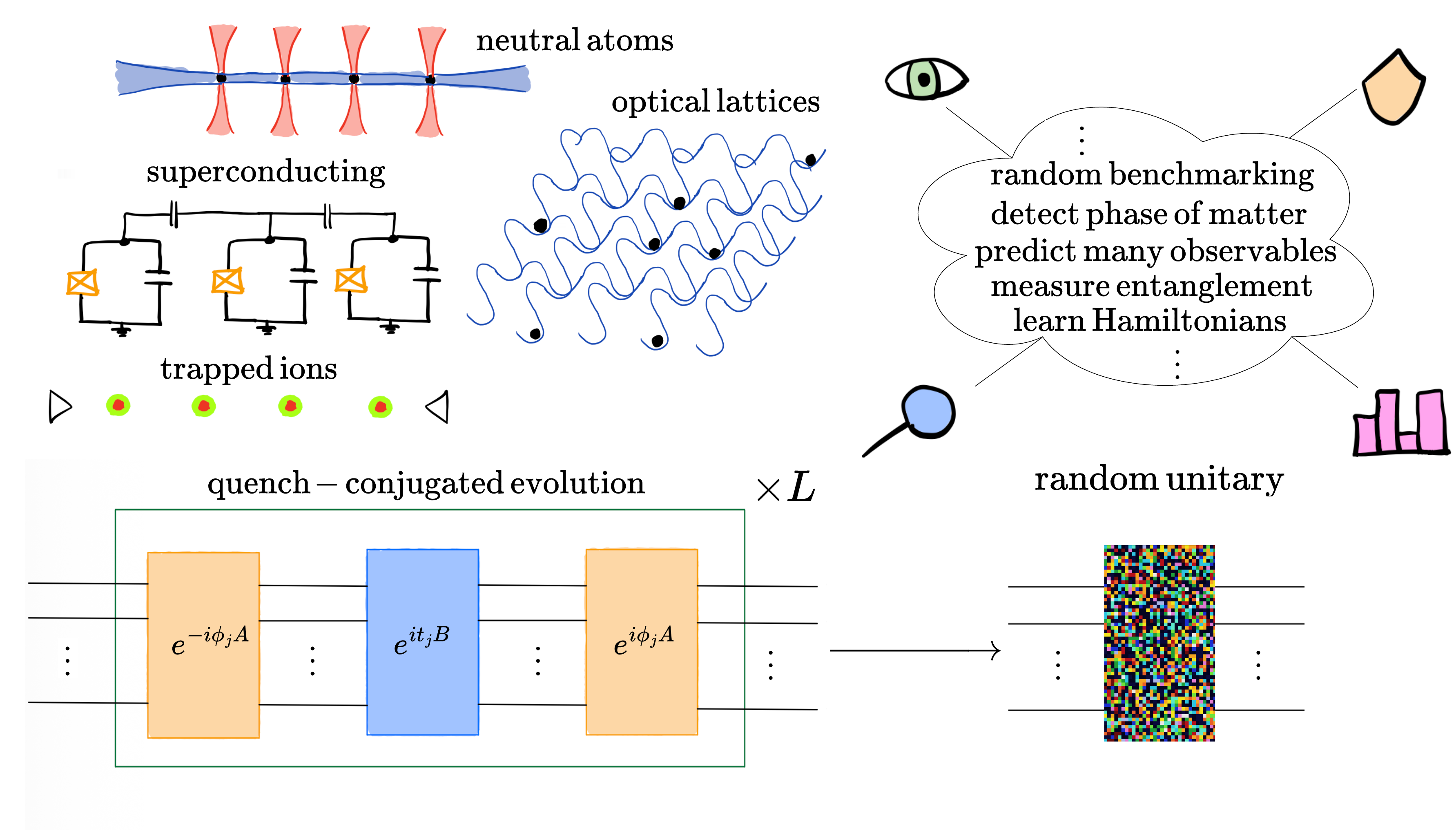}
    \caption{\textbf{Implementing random unitaries with global control.}
We introduce a protocol that applies to a broad range of experimental platforms, including neutral atoms, superconducting circuits, bosons and fermions in optical lattices, and trapped ions (upper left). Its central construction, quench-conjugated evolution (lower left), consists of $L$ successive global unitary operations, $e^{i\phi_j A}e^{it_j B}e^{-i\phi_j A}
= e^{it_j H(\phi_j)},$
where the control parameters $\phi_j$ and evolution times $t_j$ are sampled independently at each step, with the times drawn from a finite union of bounded intervals of nonnegative durations. No reversal of physical evolution is required. 
$A$ is not a separate operation that must be applied, but rather is the generator of the unitary change of basis of the Hamiltonian $B$.
Provided that the connected Lie subgroup associated with the generated Lie algebra is closed, the resulting ensemble converges exponentially in $L$ to the  Haar measure beyond a finite product-length threshold (lower right). These random unitaries provide a resource for various quantum information tasks (upper right).
}
    \label{fig:overall}
\end{figure*}
Random unitary operations are a foundational resource for quantum information processing. They underpin randomized measurements and benchmarking, with applications extending from statistical learning theory to quantum cryptography~\cite{elben2022,magesan2011,huang2020,huang2022,lancien2020,ji2018,metger2024}. Many applications require only finitely many moments of the uniform Haar distribution, reproduced by unitary designs~\cite{dankert2009,harrow2009expanders}. The challenge is to generate these statistical resources with experimentally available controls.

As quantum many-body systems grow to thousands of particles and beyond, precise, independent control of every constituent becomes increasingly demanding. Global control offers a different organizing principle for existing devices and future architectures: manipulate many particles through a small set of shared parameters. Circuit constructions provide rigorous routes to unitary designs, but their control requirements need not match those of globally driven hardware~\cite{brandao2016,metger2024}. Hamiltonian approaches bring randomization closer to native dynamics, although particular constructions require repeatedly resampling disorder in the Hamiltonian parameters during coherent evolution \cite{elben2018,vermersch2018} or invoke random-matrix assumptions about chaotic eigenvectors~\cite{zhou2026single,zhou2026three,sun2026,yang2026,nandy2026}. The latter works transfer their analytical predictions to experimentally relevant Hamiltonians primarily through numerical studies of small systems~\cite{zhou2026three,yang2026}. Their random-matrix guarantees do not, by themselves, certify randomization for a specified physical Hamiltonian. 

Global control can already support universal dynamics under suitable conditions without assumptions of quantum chaos~\cite{cesa2023,hu2025}. This suggests that useful randomness may emerge from the interplay between native dynamics and global controls. 
Lie-algebraic studies of variational circuits also establish design-convergence results for ensembles generated by sets of Hamiltonians acting on the full system, but applying these results to a prescribed control ensemble requires quantitative bounds on its moment operators~\cite{larocca2025,larocca2022,ragone2024,fontana2024}. For the global-control mechanism considered here, the challenge is therefore constructive: specify how the control parameter and evolution times should be sampled, prove that the resulting products explore the accessible group, and establish their rate of convergence to Haar measure. 
We seek these guarantees with evolution times restricted to a fixed finite union of bounded intervals of nonnegative times.\\

\paragraph*{Our result.}
We consider a quantum system with a fixed reference Hamiltonian $B$ and a global control parameter $\phi$ that rotates the Hamiltonian through unitary conjugation:
\begin{equation}\label{eqn:H_cond}
    H(\phi) = e^{i \phi A} B e^{-i \phi A},
\end{equation}
where the Hermitian operator $A$ generates the control-induced rotation. We illustrate this construction in the bottom left of \Cref{fig:overall}. 
By independently sampling the control parameter $\phi$ and evolution times, we construct a unitary random walk from successive evolutions under these Hamiltonians.
For this walk, we prove open-set irreducibility and convergence to the Haar measure, with a quantitative asymptotic rate, under a single structural condition: the connected Lie subgroup associated with the Lie algebra generated by the available Hamiltonians $\{iH(\phi)\}$ is closed. That is, we show every nonempty open region of the accessible group is reached with nonzero probability at some finite product length (``walkability''), and the endpoint distribution approaches the uniform Haar measure as the length increases (``uniformity'').\footnote{See the Technical Overview in \Cref{sec:tech_overview} for precise statements and definitions.} When the generated group is the full special unitary group, convergence reproduces the Haar moments required for approximate unitary designs, supplying the statistical resource underlying the randomized protocols described above.

Uniform randomness over the accessible group can therefore be approached using a fixed Hamiltonian and global control, without requiring additional local addressability, dynamically reconfigured disorder, or assumptions of quantum chaos. Our guarantees follow from the geometry of the product map of the conjugated evolutions, without evaluating moment operators.\\

\paragraph*{Examples.}
Our global-control mechanism is naturally available across a broad range of quantum platforms, including Rydberg atoms, bosons and fermions in optical lattices, superconducting circuits, and trapped ions, which we illustrate in the top left of \Cref{fig:overall}. We give a discussion of such implementations in \Cref{apdx:examples}. In platforms such as QuEra's Aquila Rydberg atom array~\cite{wurtz2023}, global control is a native hardware constraint. In others, including transmon systems, global driving remains available even when local control is also possible~\cite{zhao2023}. Our results thus identify both how restricted hardware can generate useful randomness and how more extensively controlled devices can access the same resource using a reduced set of controls.\\

\paragraph*{Organization.}The paper is organized as follows. 

In \Cref{sec:background}, we discuss applications of random unitaries and related work on their generation in circuit and Hamiltonian models. In \Cref{sec:tech_overview}, we formulate the technical questions and outline our proofs. In \Cref{sec:prelim_results}, we introduce the definitions and assumptions and establish the preliminary results. In \Cref{sec:main_results}, we prove open-set irreducibility and establish two routes to Haar sampling: uniform sampling of the input parameters yields exponential convergence with product length beyond a finite threshold, while, under additional hypotheses, a nonuniform sampling rule determined by the normal Jacobian and fiber volume yields exact Haar measure at finite length.

The Appendix contains five sections. \Cref{apdx:def} collects the relevant mathematical definitions, organized by subject. \Cref{apdx:nec_prior_results} restates the prior known results used to prove \Cref{thm:exp_conv}, \Cref{apdx:rep_theory} introduces the representation-theoretic concepts needed for \Cref{cor:anoussis}, and \Cref{apdx:extra_proofs} supplies our remaining technical proofs. \Cref{apdx:lin_ind} illustrates our construction of alternating conjugated Hamiltonian evolutions with one- and two-qubit examples. Finally, \Cref{apdx:examples} describes implementations in Rydberg atoms, Bose and Fermi Hubbard models, superconducting qubits, and trapped ions, and explains how symmetries constrain the accessible operations.\\

\paragraph*{Notation.}Some notational conventions we use. 
Let $n \in \mathbb{N}$. $[n] := \{0, 1, \dots, n-1 \}$.
$\dim$, unless otherwise noted, is the dimension over $\R$.
$\Lie(G)$ means the Lie algebra of the Lie group $G$.
Let $\mathfrak k$ be a Lie algebra over $\R$ and $S \subseteq \mathfrak k$. We write the generated subalgebra $\Lie_\R \langle S \rangle:= \bigcap \{ \mathfrak l \leq \mathfrak k \, : \, S \subseteq \mathfrak l\}$.
The font of a symbol indicates the type of object: fraktur i.e. mathfrak ($\g$, $\h$, $\mathfrak k$) denotes a Lie algebra; calligraphic i.e. mathcal ($\D$, $\Ucal$, $\mathcal N_L$, $\mathcal S$) a set; script i.e. mathscr ($\mathscr V_{U,k}$, $\mathscr M^{(k)}_\mu$) mixed-unitary and non-mixed-unitary channels; and sans-serif  ($\vol$, $\Vol$, $\fibvol_L$, $\Hd^{k}$) a volume. Vector spaces we write with capital Latin italics.

\section{Context and related work}\label{sec:background}
Random unitaries have a long history in quantum information and computation. Early applications used random-unitary averaging for entanglement distillation, approximate encryption, data hiding, and remote state preparation~\cite{bennett1996,hayden2004}. The prospect of reproducing useful Haar statistics with efficiently implementable operations motivated early circuit constructions and an NMR demonstration of pseudorandomness~\cite{emerson2003}, followed by the systematic development of unitary designs~\cite{dankert2009}. These ideas now underpin benchmarking, randomized measurements, and classical-shadow methods for learning quantum objects~\cite{magesan2011,brydges2019,elben2019,huang2020,elben2022}, which extract selected properties of a quantum state from random measurement statistics without reconstructing the full state. Related scrambling dynamics can generate metrologically useful entanglement for quantum-enhanced sensing~\cite{kobrin2024}. The notion of designs defined by matching finitely many Haar moments should be distinguished from cryptographic pseudorandomness, which instead requires indistinguishability by computationally bounded observers~\cite{ji2018,metger2024}.

\paragraph*{Circuits over local gate sets.} Circuit-based approaches give rigorous constructions of approximate 2-designs, establish the multiqubit Clifford group as an exact 3-design, and generate higher-order designs using quantum tensor-product expanders and local random circuits of polynomial depth~\cite{webb2016,zhu2017,dankert2009,harrow2009,harrow2009expanders,brandao2016}. Subsequent constructions have sharpened the required circuit depth and randomness~\cite{metger2024}. These results provide strong guarantees when local gate operations are available; our concern is instead how comparable randomness can arise directly from native Hamiltonian dynamics under restricted control.

\paragraph*{State designs.} A parallel line of work studies state designs and projected ensembles, in which measurements on part of a many-body system generate ensembles of conditional states on the remainder~\cite{ippoliti2022,cotler2023,choi2023,mcginley2023,iosue2024}. These studies connect quantum dynamics and measurement to the emergence of Haar-like state statistics, including through deep thermalization. Here, the distinction between states and unitaries is important: reproducing Haar statistics for a family of output states does not, by itself, establish the corresponding unitary-design property for the underlying processes.

\paragraph*{Stochastic and randomized Hamiltonian constructions.} Hamiltonian constructions are known to provide several routes to uniform randomness. Onorati et al. prove design convergence and mixing-time bounds for local stochastic Hamiltonians inducing Brownian motion on the unitary group~\cite{onorati2017}. Banchi, Burgarth, and Kastoryano consider a driven Hamiltonian \(H(t)=H_0+g(t)V\), where \(H_0\) and \(V\) are fixed and \(g(t)\) is a scalar control pulse~\cite{banchi2017}. For fully controllable systems, they establish convergence to unitary designs under Gaussian stochastic driving with stationary correlations and finite correlation time, and analyze the convergence rate in a driven chain.
Their control enters additively, whereas ours acts by conjugation and therefore preserves the instantaneous spectrum; moreover, our construction uses products of evolutions with independently sampled durations rather than continuous stochastic driving. Crucially, our result does not require full controllability: it establishes convergence in total variation to Haar measure on the closed accessible group, including when symmetries restrict it to a proper subgroup of \(U(d)\).
Nakata et al. construct designs from alternating random diagonal evolutions in complementary bases and from nearly time-independent Hamiltonians with spin-glass-type interactions~\cite{nakata2017}, with a related pseudorandom construction implemented in a 12-spin NMR system~\cite{li2019}. Together, these results establish rigorous routes from randomized Hamiltonian dynamics to unitary designs, but under control models structurally different from our setting of fixed-Hamiltonian global control.

\paragraph*{Realistic physical systems.}
For atomic simulators, Elben et al. and Vermersch et al. propose generating random unitaries through sequences of quenches in local disorder potentials, with numerical studies in spin and Hubbard models and applications to estimating density-matrix moments~\cite{elben2018,vermersch2018}. In these protocols, the disorder is resampled during coherent evolution, and each sequence must be reproduced across measurement repetitions. Closer to our control setting, Vermersch et al. also consider a reduced-control protocol in which a single fixed disorder pattern is switched on and off for randomly chosen durations, finding numerical evidence for approximate $2$-design formation in a Fermi--Hubbard model~\cite{vermersch2018}. Our results provide analytical guarantees for randomization under such restricted control, without resampling spatial disorder during the evolution. Disorder-based randomization in optical lattices itself remains a target of current experimental efforts~\cite{unirand_report}. Related Hamiltonian-shadow protocols instead use chaotic or locally scrambled dynamics to encode information about a state into measurement outcomes, constructing estimators adapted to the resulting measurement ensemble rather than requiring that ensemble to reproduce Haar statistics~\cite{hu2022,hu2023}.

\paragraph*{Chaotic Hamiltonians.} Recent work asks whether chaotic Hamiltonian dynamics, supplemented by a small number of switches or intermediate operations, can reproduce the Haar moments required for unitary designs. Proposed constructions include a single quench between independently sampled Hamiltonians~\cite{zhou2026single}, temporal ensembles of three fixed Hamiltonians~\cite{zhou2026three}, and evolutions interrupted by Pauli operations or more general perturbations~\cite{sun2026,yang2026,nandy2026}. In these approaches, time averaging randomizes phases within a fixed energy eigenbasis, while Hamiltonian switches or intermediate operations mix between eigenstates. The need for such additional operations is sharpened by Cui et al., who show that ensembles generated by time-independent, constant-local Hamiltonians remain efficiently distinguishable from Haar randomness even at arbitrarily long evolution times, and cannot form unitary $2$-designs below an inverse-superpolynomial error threshold~\cite{cui2025random}. Thus, long-time evolution under a single fixed Hamiltonian is not enough; additional control must provide a mechanism for mixing between its eigenspaces.

The Hamiltonian-based analytical guarantees rely on distinct statistical hypotheses. The single-quench calculation uses GUE spectral and eigenvector statistics~\cite{zhou2026single}, while the three-Hamiltonian temporal construction assumes nonresonant spectra, independent Haar-distributed eigenbasis overlaps, and large-dimension self-averaging~\cite{zhou2026three}. Pauli-assisted constructions invoke Gaussian matrix-element and Pauli-spectrum statistics~\cite{sun2026}, or an explicit eigenvector-thermalization ansatz~\cite{nandy2026}. Yang and Wu instead establish a Haar-eigenbasis-averaged theorem under additive nonresonance and suppression conditions on the intermediate ensemble~\cite{yang2026}. These temporal-ensemble arguments use infinite-time filtering and large-Hilbert-space asymptotics, with numerical tests in structured many-body models. Applying their design guarantees to a specified experimental Hamiltonian therefore requires establishing the relevant spectral and eigenvector hypotheses.

\paragraph*{Parametrized Hamiltonians and QAOA.} Related convergence results arise in the study of barren plateaus in variational quantum circuits, including the quantum alternating operator ansatz (QAOA). For fully controllable repeated-layer ansatze, with $\g=\mathfrak{su}(d)$ and independently and uniformly sampled gate parameters, Larocca et al. establish convergence to approximate unitary $2$-designs~\cite[Thm.~1]{larocca2022}; analogous conclusions hold within invariant subspaces when the restricted dynamics are controllable. Ragone et al. establish design-convergence bounds on the dynamical Lie group at each fixed moment order, expressed through the corresponding single-layer moment operator~\cite[Thm.~2 and Supplement Information Section III]{ragone2024}. Fontana et al. prove polynomial-depth mixing for polynomial-dimensional dynamical Lie algebras when each step selects a generator uniformly from an orthogonal basis of the algebra and samples its evolution parameter uniformly over one complete period~\cite[Thm.~2.6]{fontana2024}.

These moment-based methods provide a possible route to analyzing our ensemble. Obtaining quantitative depth bounds, however, requires bounding the largest singular value of the single-step moment operator restricted to the orthogonal complement of its Haar-invariant subspace. In the chaotic-Hamiltonian constructions discussed above, random-matrix eigenbasis statistics make related moment and frame-potential calculations tractable. Obtaining comparable bounds for a specified Hamiltonian without such statistical assumptions remains a nontrivial analytical problem.\\

Our approach bypasses moment-operator analysis by studying the geometry of the map that generates the unitary random walk. 
Motivated by the use of global-phase quenches on QuEra's Aquila processor~\cite{matsoukasroubeas2026} to measure Renyi entropy, we consider a general framework interleaved evolutions of a fixed Hamiltonian under global conjugation. 
Our central technical contribution is an explicit construction of a submersion point, with a product-length bound determined by the dimension of the generated Lie algebra and the number of linearly independent conjugated Hamiltonians~\Cref{lem:L0_bound}. If the connected Lie subgroup generated by $\{iH(\phi)\}$ is closed, this construction ensures that the random walk's endpoint distribution becomes absolutely continuous with respect to Haar measure on the accessible group after finitely many steps~\Cref{thm:non0_ac}. This implies open-set irreducibility: every nonempty open region of the accessible group can be reached with positive probability after finitely many steps. 
Established compact-group random-walk results~\Cref{apdx:nec_prior_results} then allow us to show exponential convergence in total variation to Haar measure. 
Determining the sequence lengths and physical evolution times required for specific hardware and randomization tasks remains an open quantitative problem.

\section{Technical overview}\label{sec:tech_overview}
\begin{figure*}[ht]
\centering
\begin{tikzpicture}[
    font=\small,
    >={Stealth[length=2mm]},
    arr/.style={->, thick},
    lbl/.style={inner sep=2pt},
    box/.style={draw, thick, rounded corners=1pt, inner sep=4pt},
    squigbox/.style={draw, thick, decorate, decoration={snake, amplitude=0.5pt, segment length=4pt}, inner sep=4pt},
]

\node[box, dgYellow] (MR1) at (2.4,5.7)  {Conv. to Haar:~\Cref{thm:exp_conv}};
\node[box, dgYellow] (MR2) at (2.7,4.7)  {Exact Haar: \Cref{thm:exact}};
\node[lbl, dgCyan2]   (BAG) at (7.8,5.7)  {\cite[Thm.~3]{bhattacharya1972}, \cite[Cor.~4.2]{anoussis2004}};
\node[squigbox, dgPurple] (AC)  at (12.0,5.2) {A.c.: \Cref{thm:non0_ac}};
\node[lbl, dgGreen]  (GS)  at (12.0,4.1) {\Cref{lem:L0_bound}};
\node[box, dgPurple] (Fin) at (1.4,6.8) {Unitary designs \Cref{cor:designs}};
\node[lbl, dgPurple] (Off) at (10.0,6.5) {OS. irred. \Cref{cor:os_irred}};

\node[dgViolet, font=\normalsize] at (0.25,4.0) {Main \Cref{sec:main_results}};
\draw[dgViolet, thick] (-0.9,3.55) -- (15.3,3.4);
\node[dgViolet, font=\normalsize, anchor=west] at (-1.05,3.12) {Prelim. \Cref{sec:prelim_results}};

\node[lbl]                       (Def) at (0.1,1.4)  {\Cref{def:main}};
\node[lbl, dgOrange, anchor=west] (G)   at (2.0,2.2)  {\Cref{prop:G_group}};
\node[lbl, dgGreen,  anchor=west] (J)   at (2.0,1.4)  {\Cref{def:normal_jacob}, \Cref{rem:jacobian}};
\node[lbl, dgBrown,  anchor=west] (h)   at (2.0,0.6)  {\Cref{eqn:h_subalgebra}, \Cref{lem:V_m_h}};
\node[lbl, dgPurple]             (mu)  at (6.3,3.0)  {\Cref{def:pushforward}};
\node[lbl, dgGreen,  anchor=west] (Sub) at (7.0,1.4)  {\Cref{lem:jacobian_real_analytic}, \Cref{prop:submersion_almost_everywhere}};
\node[lbl, dgCyan,   anchor=west] (im)  at ([xshift=0.8cm, yshift=-0.4cm]Sub.east) {\Cref{lem:im_J_subset_h}};
\node[lbl, dgOrange]             (Con) at (14.0,2.2) {\Cref{prop:h_g_submersion}};
\node[lbl, dgOrange]             (N)   at (7.8,2.2) {\Cref{prop:subalgebra}};
\node[lbl, anchor=west]          (Ext) at (12.9,.5) {\Cref{lem:comm_comm_sufficient}};

\coordinate (D) at ([xshift=2pt]Def.east);
\coordinate (Gout) at ([xshift=2pt]G.east);
\coordinate (Nout) at ([xshift=2pt]N.east);
\coordinate (Sout) at (11.3,1.4);
\coordinate (Aout) at ([xshift=-2pt]AC.west);
\coordinate(hout) at (5.33,0.6);
\draw[arr, dgOrange] (D) to[out=30,in=195]  (G.west);
\draw[arr, dgGreen]  (D) -- (J.west);
\draw[arr, dgBrown]  (D) to[out=-30,in=165] (h.west);
\draw[arr, dgPurple] (Gout) to[out=30,in=230] (mu);
\draw[arr, dgPurple] (mu.east) to[out=15,in=230] ([xshift=-9pt]AC.south);
\draw[arr, dgOrange] (Gout) -- (N);
\draw[arr, dgOrange] (Nout) -- (Con);
\draw[arr, dgOrange]  (hout) to[out=10,in=210] (N);
\draw[arr, dgGreen]  (Sout) to[out=70,in=235] ([xshift=-4pt]GS.south);
\draw[arr, dgGreen]  (J) -- (Sub);

\draw[arr, dgCyan]   (Sout) to[out=-45,in=180] (im.west);
\draw[arr, dgGreen]  (Nout) .. controls ([xshift=1.1cm, yshift=-0.3cm]Nout) and ([xshift=-0.25cm, yshift=0.8cm]im.north) .. (im.north);

\draw[arr, dgOrange] (im.east) to[out=15,in=290] (Con.south);
\draw[arr, dashed]   (Ext) to[out=30,in=-45] (Con);

\draw[arr, dgGreen]  (Con.north) to[out=95,in=-80] ([xshift=4pt]GS.south);
\draw[arr, dgPurple] (GS) -- (AC);
\draw[arr, dgCyan2]  (Aout) to[out=180,in=-30] (BAG.south);
\draw[arr, dgYellow] (BAG) -- (MR1);
\draw[arr, dgYellow] (Aout) to[out=210,in=0] (MR2.east);
\draw[arr, dgPurple] (MR1.north -| Fin) -- (Fin);
\draw[arr, dgPurple] (MR2.west) .. controls (-0.6,4.7) and (-0.6,6.3) .. (Fin.south);
\draw[arr, dgPurple] (Aout) to[out=120,in=270] (Off.south);

\end{tikzpicture}
\caption{\textbf{Proof flow.} Fundamental definitions appear in black at the lower left; arrows indicate logical dependencies moving to the right, up, and back to the left. Colors group results by topic: orange for the group $G$; green for the Jacobian of $U_L$; brown for the generated Lie subalgebra $\h$; dark purple for measure-theoretic properties of the pushforward measure $\mu$; and light blue for the results of Bhattacharya and Anoussis--Gatzouras, restated in \Cref{apdx:nec_prior_results}. 
The box with a wavy border (absolute continuity, a.c., with respect to the Haar measure) implies open-set irreducibility (OS. irred.), while the boxes with straight borders establish convergence to the Haar measure. The exact-Haar branch additionally requires the fiber-volume and boundedness hypotheses of \Cref{thm:exact}.
The dashed black line indicates a sufficient condition for $\h=\g$, the structural assumption used to construct a finite-length submersion point.
The purple line separates the preliminary results in \Cref{sec:prelim_results} from the main results in \Cref{sec:main_results}.}
\label{fig:proof}
\end{figure*}
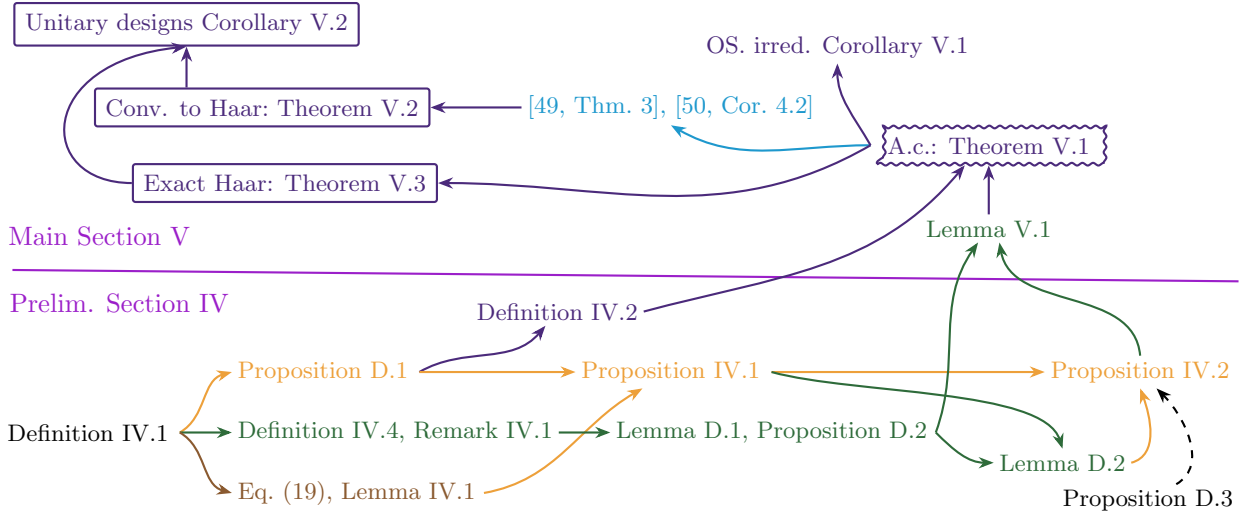

\begin{figure*}[!ht]
    \includegraphics[width=\linewidth]{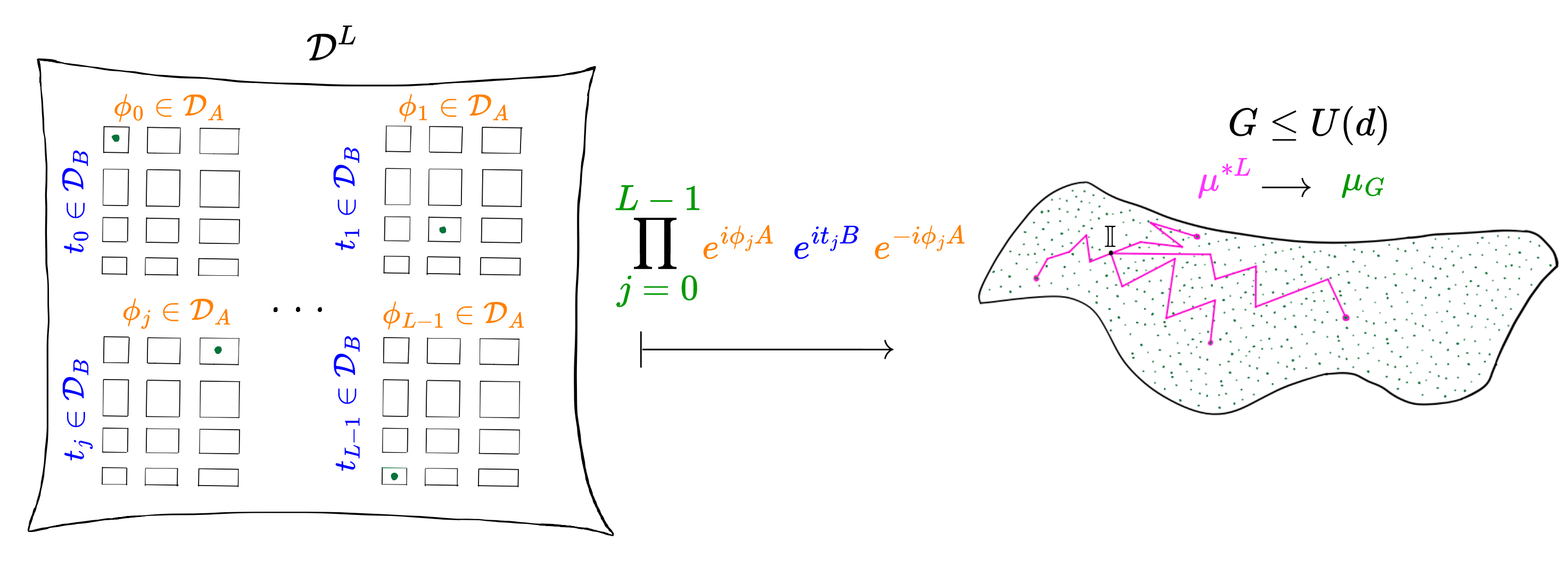}
    \caption{\textbf{Schematic of the protocol.}
Left: each collection of $L$ parameter pairs $(\phi_j,\timeB_j)\in\D_A\times\D_B$, with $\phi_j$ shown in orange and $\timeB_j$ in blue, defines a point $x\in\D^L$ (green).
Middle: the product map
$U_L(x)=\prod_{j=0}^{L-1} e^{i\phi_j A} e^{i\timeB_j B} e^{-i\phi_j A}$
maps these parameters to a unitary in the compact Lie group $G\leq \mathrm{U}(d)$. Each factor corresponds to one step of the random walk (magenta). For independent, identically distributed steps with distribution $\mu$, the endpoint distribution after $L$ steps is $\mu^{*L}$. Our main result establishes conditions under which this distribution converges to the Haar measure $\mu_G$, the uniform probability measure on $G$, illustrated by the green points on the right.}
    \label{fig:schematic}
\end{figure*}

Consider a system with a Hamiltonian and a parameterized global control operation that acts rapidly compared with the system's dynamics. 
\begin{question}[Open-set irreducibility]\label{question:topo_irred}
    Is it possible to choose the Hamiltonian and the sampling of the global control parameter so that the resulting random walk on the compact group generated by these operations is \emph{open-set irreducible}?\footnote{See \Cref{def:topo_irred}.}. That is, starting at $\I$, can a walk reach every nonempty open subset of the group with positive probability in finitely many steps? 
\end{question}

\begin{question}[Asymptotic uniformity]\label{question:conv}
    Is it possible to choose the Hamiltonian and the sampling of the global control parameter so that the resulting random walk on the compact group generated by these operations \emph{converges in total variation to Haar measure}?\footnote{See \Cref{def:asymptotic_uniformity}.}. That is, starting at $\I$, can the total variation distance between the distribution of endpoints from the walk and Haar measure tend to zero as the number of steps goes to infinity?
\end{question}

\paragraph*{Results.}
We establish both properties under the conditions stated below. Suppose that the system is governed by a $d\times d$ Hamiltonian $H(\phi)$ depending on a control parameter $\phi$, with fixed Hermitian matrices $A$ and $B$ such that $H(\phi)=e^{i\phi A}Be^{-i\phi A}$, as in \Cref{eqn:H_cond}. Our construction starts from the \emph{quench-conjugated Hamiltonian evolution}
\begin{equation}\label{eqn:F_first}
    F(\phi,\timeB) := e^{i\phi A}e^{i\timeB B}e^{-i\phi A}.
\end{equation}
Here, $B$ is the system's reference Hamiltonian, while $A$ generates the unitary change of basis induced by the control parameter. We construct products of $F$ with independently sampled control parameters and evolution times. Each additional factor defines a step of a unitary random walk on the group generated by the allowed evolutions. This is illustrated in \Cref{fig:schematic}.

\subsection{Proof overview: high-level}
In \Cref{sec:preliminaries}, we establish the sufficient conditions that imply open-set irreducibility and asymptotic uniformity.
In particular, the choice of $F$ enables us to analytically prove the properties of the total derivative of the products of $F$.
This key step, \Cref{lem:L0_bound}, is our main technical contribution.
We elaborate in the next subsection.

In \Cref{sec:main_results}, we establish two routes to Haar sampling. The first is to recognize that our first stage results allow us to apply results Bhattacharya~\cite[Thm.~3]{bhattacharya1972} and Anoussis--Gatzouras~\cite[Cor.~4.2]{anoussis2004}. These allow us to choose a uniform measure on the input parameter domain  and obtain an exponentially shrinking distance, with an explicit form of the asymptotic rate, from our walk endpoint distribution to the uniform Haar distribution as the number of steps in the random walk increases.
Our second approach specifies a joint distribution on the parameter pairs that yields exact Haar measure at finite length, provided the regular fiber volume is positive and finite for Haar-almost every output and the ratio of the normal Jacobian to the fiber volume is essentially bounded with respect to the uniform input measure.

\subsection{Proof overview: detailed}\label{sec:proof_overview}

To understand the argument more deeply, we should have in mind the results of Bhattacharya~\cite[Thm.~3]{bhattacharya1972} and Anoussis--Gatzouras~\cite[Cor.~4.2]{anoussis2004}, which we provide in Appendix \Cref{apdx:nec_prior_results}.

Bhattacharya's result says the following: suppose we choose elements independently from a compact, connected group $G$ described by a probability measure $\mu$.
We want to show that as we choose more elements from $G$ and combine them together using the group operation $L$ times, the resulting measure $\mu^{*L}$ converges to the (normalized) uniform random measure, which is called the (normalized) Haar measure $\mu_G$.
Bhattacharya says that if for some $L$,  $\mu^{*L}$ has a non-zero absolutely continuous component with respect to $\mu_G$,\footnote{See \Cref{def:a.c.}.} then for $L' \geq L$ the total variation distance of $\mu^{*L'}$ to $\mu_G$ decays exponentially fast in $L'-L$.

In light of these results, it remains to establish the following statements:
(1) show $G$ is compact and connected, and 
(2) bound a finite length $L$ for which $\mu^{*L}\ll\mu_G$.

The second guiding result, due to Anoussis--Gatzouras, takes Bhattacharya's theorem a step further by providing an explicit form of the asymptotic rate of convergence using representation theory. Therefore, once we have shown (1) and (2), the asymptotic rate is implied. This is how we obtain \Cref{thm:exp_conv}.
There is also an explicit, but more conservative, rate that we can obtain using Doeblin minorization~\cite[Thm.~8]{roberts2004},  that we address in \Cref{rem:doeblin}.

Once we have shown statements (1) and (2), we show in \Cref{thm:exact} how modifying the joint distribution of the input parameters yields exact Haar measure at finite length. This requires two additional hypotheses: the regular fiber volume is positive and finite for Haar-almost every output, and the ratio of the normal Jacobian to the fiber volume is essentially bounded with respect to the uniform input measure. The idea is to eliminate bias by adjusting the input sampling to compensate directly for variations in the normal Jacobian and fiber volume.

There are several necessary results for the main argument, but their proofs distract from the important flow of ideas. To keep the exposition focused, we collect these proofs in Appendix \Cref{apdx:extra_proofs}.

\Cref{fig:proof} shows the overall flow of ideas. We begin on the bottom left, flow upward and to the right, and then wrap back around to the upper left.
A purple line divides the ``Preliminaries" in the bottom with the ``Main'' results in the upper part.
Arrows denote the flow of dependencies.
In the lower left, the base definitions (\Cref{def:main}) along with the specification of $A$ and $B$ are shown in black.  
Three parallel directions open: showing $G$ is a compact Lie subgroup (orange), deriving an explicit form of the Jacobian of the length $L$ product map $U_L(x)= \prod_{j=0}^{L-1}F(x_j)$ (green), and constructing the  generated Lie subalgebra $\h$ from a linearly independent set of the generators of $F(x)$ (brown).\footnote{See \Cref{def:gen_lie_sub}.}

\Cref{prop:G_group} is the first result of the orange branch in \Cref{fig:proof}. This gives $G$ the compactness needed in statement (1), but not yet the connectedness. Connectedness will be automatically resolved in \Cref{prop:subalgebra}.
Moreover, this proof does not require the time domain $\D$ to contain negative durations, which could present a physical problem of ``evolving backwards in time'', or equivalently having to flip the sign of every term in the Hamiltonian, which may not be possible in a given physical system.

We choose a construction of $\h$ in \Cref{lem:V_m_h} to match the form of the columns of the Jacobian $J_{L,x}$. We explicitly link the two via \Cref{lem:im_J_subset_h} (green), by feeding \Cref{lem:V_m_h} into \Cref{prop:subalgebra} (orange) which we elaborate on below, in which we show that $\mathrm{im}(J_{L,x}) \subseteq \h$. 
Continuing in this line of argument, we show in \Cref{lem:jacobian_real_analytic} and \Cref{prop:submersion_almost_everywhere} (green) that once we can demonstrate a single submersion point, i.e. a choice of input parameters such that the Jacobian has full rank, then $U_L$ is a submersion almost everywhere. 

Keeping in mind that our goal is to make measure theoretic statements, we define the pushforward measure $\mu$ and the normalized Haar measure on $G$, $\mu_G$ (dark purple) branching off from the initial result about $G$.

Once a submersion point exists at length $L$, real analyticity implies that $U_L$ is a submersion almost everywhere.
The coarea formula\footnote{See \Cref{thm:coarea}.} then gives $\mu^{*L}\ll\mu_G$, completing our second technical task.
The intuition is that at a submersion point, a walker can move in all tangent directions of $G$. If $G$ is compact and connected, then this local reachability extends, through repeated steps, to the ability to reach every nonempty open subset of $G$ with positive probability, which is open-set irreducibility.

Next, we resolve connectedness.
In \Cref{prop:subalgebra} (orange), we show three main results. First, $\h$ is a Lie subalgebra of $\g = \mathrm{Lie}(G)$. Then $\h$ has a corresponding unique, connected immersed Lie group $H_\h$ by the Lie subalgebra-Lie group correspondence (\Cref{def:subalgebra}). 
Second, we show $\overline{H_\h} = G$. This implies $G$ is connected.
Third, $\h = \g$ if and only if $H_\h$ is closed. 
From \Cref{prop:h_g_submersion} (green), a necessary condition for a submersion point to exist is $\h = \g$.
To determine, given $A$ and $B$, whether $\h = \g$, we give a sufficient condition in \Cref{lem:comm_comm_sufficient}. It is not required for any of the main results, but it is an easy test to check.
Our first technical main question is now complete.

The first of our major results is to construct an explicit submersion point in finite length, which we do in \Cref{lem:L0_bound} (green). In particular, we obtain a bound on the minimum $L_0$ to yield a submersion point if $\h = \g$. This is relatively straightforward because of the structure of $F$.

In particular, $F(\phi,t)=e^{i\phi A}e^{itB}e^{-i\phi A}=e^{tC_\phi}$
satisfies $F(\phi,0)=I$ for every $\phi$. 
Consequently, when all time variables vanish, the time derivative columns of the product Jacobian are $C(\phi_j)=iH(\phi_j)$. By the definition of $r$ in \Cref{eqn:r}, we can choose $r$ linearly independent such columns.
Then, since $V_{m_0}=\g$, for any current product $P$, the vectors $\Ad_{PU_\ell(\xi)}C_\phi$, with $\ell\leq m_0$, span $\g$. Unless the Jacobian already has full rank, we can therefore choose one outside the span of its existing columns and append the $\ell$ conjugating factors followed by $F(\phi,0)$ to the point that will become a submersion point. Repeating this construction adds at least one independent column using at most $m_0+1$ additional factors each time, until full rank is reached. 
Finally, continuity of the differential allows us to replace the zero durations by sufficiently small positive ones while preserving full rank, yielding an interior submersion point at the same product length.

In contrast, if we chose instead the unconjugated form $O(\phi,t)=e^{i\phi A}e^{itB}$, residual evolutions under $A$ remain even at zero times and conjugate subsequent Jacobian columns, making this construction less direct and without, it would appear, a straightforward way to construct an explicit submersion point without imposing more structural constraints on $A$ and $B$.

At this point, the results of the next section follow. \Cref{thm:non0_ac} (dark purple) requires a submersion point almost everywhere, which we have just supplied, and uses the coarea formula as observed before. 
This implies open-set irreducibility, which we formalize in \Cref{cor:os_irred}, which answers \Cref{question:topo_irred}. 
\Cref{thm:exp_conv} establishes exponential convergence, while \Cref{thm:exact} gives exact Haar sampling under its additional hypotheses. Both address \Cref{question:conv}. When $G=\mathrm{U}(d)$, these results yield the corresponding unitary-design statements in \Cref{cor:designs}.

\section{Preliminary Results}\label{sec:prelim_results}

In this section, we prove the results in the bottom half of \Cref{fig:proof}. We begin in the bottom left corner in black by defining the main objects of our study.

We include zero evolution time for use in the construction below, but an open interval around zero would also include negative times; we therefore keep the allowed times nonnegative, use an extension when taking derivatives at zero, and later replace zero durations by sufficiently small positive ones.

\subsection{Our preliminaries}\label{sec:preliminaries}
Let $A,B$ be $d\times d$ Hermitian matrices and take $B \neq 0$.
Let $\D_A\subset\R$ be a nonempty finite union of bounded open intervals. Let $\D_B\subset[0,\infty)$ be a bounded, relatively open set that is a finite union of intervals, and assume that
\begin{equation}
    [0,\varepsilon)\subseteq\D_B
\end{equation}
for some $\varepsilon>0$.
The other components are bounded open intervals. Neither parameter domain is required to be connected.

Write $\D=\D_A\times\D_B$ for the set of allowed parameter pairs $(\phi,t)$. Its interior consists of the pairs with strictly positive allowed times:
\begin{equation}
    \D^\circ=\D_A\times(\D_B\cap(0,\infty)).
\end{equation}
Because the parameter domains are bounded and contain intervals of positive length, $\D$ has positive, finite Lebesgue area. We therefore define the uniform probability measure on $\D$ by\footnote{See \Cref{def:sigma_algebra}.}
\begin{equation}
    \lambda(E)
    :=\frac{\Leb(E)}{\Leb(\D)},
    \qquad E\in\mathcal B(\D).
\end{equation}

\begin{definition}[Main objects]\label{def:main}
Let $\phi \in \D_A$,  $\timeB\in \D_B$, and $\D=\D_A\times\D_B$. 
Write $\Ad_X Y = X Y X^{-1}$.
We express \Cref{eqn:F_first} in terms of its generator explicitly, 
\begin{equation}
\begin{gathered}
    F(\phi, \timeB) := e^{\timeB C(\phi)}, \quad C(\phi) = i \Ad_{e^{i \phi A}} B = i H(\phi), \\ F \, : \, \D \longrightarrow \mathrm{U}(d).
\end{gathered}
\end{equation}
Write $x = (x_0, x_1, \dots, x_{L-1})$ where each $x_j = (\phi_j, \timeB_j)\in\D$. Define
\begin{equation}
    U_L(x) := \prod_{j=0}^{L-1} F(x_j), \quad \quad  U_L \, : \, \D^L \longrightarrow \mathrm{U}(d).
\end{equation}
For $L\geq1$, we write $\Ucal_L$ for the set of all length $L$ products of time evolutions under conjugates of the fixed Hamiltonian, i.e. the image $U_L(\D^L)\subseteq \mathrm{U}(d)$. 

Define
\begin{equation}\label{eqn:G}
    G := \overline{\Ucal}, \quad\text{where}\quad \Ucal := \bigcup_{L \geq 1} \Ucal_L.
\end{equation}
We take $n = \dim G$, and note $\dim \D^L = 2L$. 
\end{definition}

The first result we show from the definitions in the orange branch in \Cref{fig:proof} is the structure of $G$.
As we show in the Appendix in \Cref{prop:G_group},  $G$ is a compact Lie subgroup of $\mathrm{U}(d)$.
Denote the Lie algebra of $G$ by 
\begin{equation}\label{eqn:frakg}
    \g := \Lie(G) = T_\I G,
\end{equation}
where $T_\I$ denotes the tangent space of $G$ at $\I$.
Since $G$ is compact, it has a unique normalized Haar measure we write as $\mu_G$ (\Cref{def:haar_G}).

Since $G$ is compact (Hausdorff), it has a unique normalized uniform (Haar) measure. This is because intuitively, the Haar measure assigns finite probability mass to sufficiently small neighborhoods. Compactness implies that finitely many such neighborhoods cover the whole group, so the total mass is finite. Therefore it is possible to define a measure that has a total mass $1$. In contrast, the Lebesgue measure on $\R$, which is not compact, has infinite mass. 
We define the resulting probability measure on $G$.

\begin{definition}[Pushforward measure on $G$.]\footnote{See \cite[Thm.~5.22]{shapiro2025measure}, also \cite{androma:change-of-variables-formula-for-pushforward-measures}}\label{def:pushforward}
Let $\lambda$ be a probability measure on $\D$, and let $F:\D\to G$ be measurable. The pushforward measure $\mu=F_{\#}\lambda$ on $G$ is defined by
\begin{equation}
    \mu(E):=\lambda\bigl(F^{-1}(E)\bigr),  \;
    E\in\mathcal B(G).
\end{equation}
$\mathcal B(G)$ is the Borel $\sigma$-algebra of $G$.
The $L$-fold convolution of $\mu$ satisfies
\begin{equation}
    \mu^{*L}=(U_L)_{\#}(\lambda^{\otimes L}),
\end{equation}
where $U_L(x) = \prod_{j=0}^{L-1} F_j(x_j)$.
\end{definition}

Throughout this work, we take $\lambda$ to be the normalized Lebesgue (uniform) measure on $\D$, so that $\lambda(\D)=1$.

Here, $F$ is continuous, because the matrix exponential is continuous, and has $F(\D) \subseteq G$. The preimage of an open set under a continuous function is open, and taking the preimage preserve complements and countable union, hence for every $E \in \mathcal B(G)$
\begin{gather*}
    F^{-1}(E)\in\mathcal B(\D),
\end{gather*}
so $F$ is a Borel measurable function.

With these definitions and initial result, we can already understand the broad picture of our scheme, illustrated in \Cref{fig:schematic}.
Sampling $x_0,\ldots,x_{L-1}$ independently from $\lambda$, shown as green points, produces the evolution operators $F(x_1),\ldots,F(x_L)$ with distribution $\mu$. 
Their product $U_L(x_1,\ldots,x_L)$ is the endpoint after $L$ steps of a random walk on $G$ starting at the identity, shown as the magenta line. 
We want to show that as $L$ increases, the distribution of endpoints described by the measure $\mu^{*L}$ approaches the uniform random distribution $\mu_G$, which we have shown in green dots on $G$.

We begin by finding the Jacobian of $U_L$ in order to characterize when its total derivative has full rank, which is shown as the green branch of \Cref{fig:proof}, Specifically, we define and study the normal Jacobian. To understand this geometrically, imagine a map between two topological manifolds. A topological manifold is a space that, near every point, looks like ordinary Euclidean space of a fixed dimension. For example, a sphere and a coffee mug are both topological manifolds.
Moving along a level set leaves the output of the map unchanged, while moving perpendicular to it changes the output. The normal Jacobian measures how much the map stretches or shrinks volume in these perpendicular directions.

For differential calculations, introduce the open set
\[
    \widehat\D:=\D_A\times\R
\]
and extend $U_L$ to $\widehat\D^L$ by the same product formula with real-valued times:
\[
    \widetilde U_L(x)
    :=\prod_{j=0}^{L-1}e^{t_jC(\phi_j)}.
\]
This extension is real analytic and agrees with $U_L$ on $\D^L$. It takes values in $G$: for each $\phi\in\D_A$, the inclusion $[0,\varepsilon)\subseteq\D_B$ and the group property of $G$ imply that $e^{tC(\phi)}\in G$ for every $t\in\R$.

We regard $\D^L$ as a manifold with corners and identify its tangent spaces with $\R^{2L}$. At boundary points, all derivatives below are computed using $\widetilde U_L$.

\begin{definition}[Right-trivialized differential]
For $x\in\D^L$, define
\begin{equation}\label{def:right_trivialized_differential}
\begin{gathered}
    J_{L,x}:\R^{2L}\longrightarrow\g,\\
    J_{L,x}(v)
    :=D\widetilde U_L(x)[v]\,\widetilde U_L(x)^{-1}.
\end{gathered}
\end{equation}
For $x\in(\D^\circ)^L$, this is the ordinary right-trivialized differential of $U_L$.
\end{definition}

\begin{definition}[Normal Jacobian]\label{def:normal_jacob}
We equip $\g$ with the real Hilbert--Schmidt inner product
\[
    \langle X,Y\rangle
    :=\operatorname{Re}\tr(X^\dagger Y),
\]
which induces a bi-invariant (invariant under left or right multiplication) Riemannian metric on $G$. We represent $J_{L,x}$ using the standard Euclidean basis in the parameter space and an orthonormal basis of $\g$ for this inner product, so that $J_{L,x}^{T}$ represents the adjoint of the differential.
Define
\begin{equation}\label{eqn:normal_jacobian}
    \Lambda_L(x) := \det \!\left(J_{L,x}J_{L,x}^{T}\right).
\end{equation}
as the squared normal Jacobian of $U_L$ at $x$. Hence $\sqrt{\Lambda_L(x)}$ is the normal Jacobian of $U_L$ at $x$.
The differential $J_{L,x}$ is surjective if and only if $\Lambda_L(x)>0$.
We call $x$ a submersion point of $U_L$ when $x\in(\D^\circ)^L$ and $\Lambda_L(x)>0$.
At boundary points, $\Lambda_L(x)>0$ instead describes full rank of the extended map $\widetilde U_L$.
\end{definition}

The Jacobian represents the differential as a real matrix. The differential is therefore surjective if and only if the Jacobian has full row rank.

Having defined a submersion point and given an explicit expression for the Jacobian, we continue along the green branch of \Cref{fig:proof} and prove that if $U_L$ is a submersion at one point in $\D^L$, then it is a submersion almost everywhere. We use the real analytic zero set \Cref{prop:real_analytic_zero} to extend statements about a single point to almost every point in $\D^L$.

\begin{remark}[Explicit coordinate representation of the Jacobian]\label{rem:jacobian}
For $x_j=(\phi_j,\timeB_j)$, directly taking the derivative yields
\begin{align}
    \frac{\partial F(\phi_j,\timeB_j)}{\partial \timeB_j}F(\phi_j,\timeB_j)^{-1}
    &=C(\phi_j),
    \label{eqn:partial_t}\\
    \frac{\partial F(\phi_j,\timeB_j)}{\partial\phi_j}F(\phi_j,\timeB_j)^{-1}
    &=
    iA-i\Ad_{e^{i\phi_j A}}\Ad_{e^{i\timeB_j B}}A 
    \label{eqn:partial_phi}
\end{align}
Define
\begin{equation*}
\Lprod{j}:=\prod_{k=0}^{j-1}F(x_k), \; \; \;
   \Rprod{j}:=\prod_{k=j+1}^{L-1}F(x_k),
\end{equation*}
where an empty product is taken to be $\I$. Since
$U_L(x)=\Lprod{j}F(x_j) \Rprod{j}$, we obtain
\begin{gather*}
    J_{L,x}\left(\frac{\partial}{\partial \theta_j}\right)
    =
    \Ad_{\Lprod{j}}\left(
        \frac{\partial F(x_j)}{\partial \theta_j}
        F(x_j)^{-1}
    \right),
    \\
    \theta_j\in\{\phi_j,\timeB_j\}.
\end{gather*}
Explicitly, substituting in ~\Cref{eqn:partial_phi}:
\begin{equation}\label{eqn:jacobian_formula}
\begin{gathered}
     J_{L,x}\left(\frac{\partial}{\partial \timeB_j}\right)
    =\Ad_{\Lprod{j}}C(\phi_j),
\end{gathered}
\end{equation}
\end{remark}

In \Cref{fig:proof}, we want to connect the orange results about $G$ with the green discussion about the Jacobian. The aim is to construct a specific submerison point.
To do this, we introduce a new branch from the root definitions: characterizing the generated Lie subalgebra $\h$ from the set of linearly independent generators of $F$, shown in brown.

Write the set 
\begin{equation}\label{eqn:Scal}
    \mathcal S =  \{ C(\phi) \, : \, \phi \in \D_A \},
\end{equation}
define 
\begin{equation}\label{eqn:r}
    r := \dim\Span_\R \mathcal S,
\end{equation}
and define the generated Lie subalgebra\footnote{See \Cref{def:gen_lie_sub}.}
\begin{equation}\label{eqn:h_subalgebra}
    \h := \Lie_\R  \langle \mathcal S \rangle.
\end{equation}

We begin connecting the three early branches in \Cref{fig:proof}. We first connect the brown Lie subalgebra branch and orange $G$ branch by relating $\h$ to $\g$ through \Cref{prop:subalgebra}. 
In the Appendix, we show \Cref{lem:im_J_subset_h} which allows us to connect these branches by relating the image of columns of the Jacobian to $\h$. Through these connections, we obtain a necessary condition to achieve a submersion point of $U_L$ in \Cref{prop:h_g_submersion}.

\begin{definition}[Span of conjugated generators]
\label{def:V_m}
Recall $\mathcal S$ from \Cref{eqn:Scal}.
For $m\in\mathbb N_0$, define
\begin{multline}\label{eqn:V_m}
    V_m
    :=
    \Span_{\mathbb R}
    \Bigl\{
        \Ad_{U_\ell(\xi)}C(\phi): \\
        \phi\in\D_A,\;
        0\leq\ell\leq m,\;
        \xi\in
        \bigl(\D_A\times[0,\varepsilon)\bigr)^\ell
    \Bigr\}.
\end{multline}
Define $U_0(\varnothing)=\I$. Consequently,
\begin{gather*}
    V_0
    =\Span_{\mathbb R}\mathcal S,
    \quad
    \dim V_0=r.
\end{gather*}
\end{definition}

\begin{lemma}[Explicit form of $\h$]
\label{lem:V_m_h}
Let $(V_m)_{m\geq0}$ be as in \Cref{def:V_m}. There exists a finite
\begin{equation}\label{eqn:m0_def}
    m_0 := \min\{m\geq0:V_m=V_{m+1}\}
\end{equation}
satisfying
\begin{gather*}
    m_0\leq\dim\mathfrak h-r\leq d^2-r.
\end{gather*}
Moreover,
\begin{gather*}
    V_{m_0}=\mathfrak h.
\end{gather*}
\end{lemma}
\begin{proof}
Since $\h$ is a Lie subalgebra containing $C(\phi)$ for $\phi \in \D_A$, it is closed under taking brackets with $C(\phi)$:
\begin{gather*}
\ad_{C(\phi)}(\h)\subseteq \h \quad \Longrightarrow \quad e^{\timeB\ad_{C(\phi)}}(\h) \subseteq \h,
\end{gather*}
since $\h$ is finite dimensional.
By Campbell's lemma (\Cref{lem:campbell}),
\begin{gather*}
\Ad_{F(\phi,\timeB)} =\Ad_{e^{\timeB C(\phi)}} =e^{\timeB\ad_{C(\phi)}} \\ \Longrightarrow \quad \Ad_{F(\phi,\timeB)}(\h) \subseteq \h,
\end{gather*}
Since $\Ad_{AB}=\Ad_A \Ad_B$, applying the prior result iteratively yields
\begin{gather*}
\Ad_{U_\ell(\xi)} (\h) \subseteq \h,
\end{gather*}
for every choice of $\xi \in (\D_A \times [0, \varepsilon))^{\ell}$ appearing in the definition of $V_m$, including trivially at $\ell=0$. 
In particular, since $ C(\phi) \in \h $,
\begin{gather*}
\Ad_{U_\ell(\xi)}C(\phi)\in\mathfrak h.
\end{gather*}
Taking the span over $\R$ of $\Ad_{U_\ell(\xi)}C(\phi)$ gives
\begin{gather*}
V_m\subseteq\h.
\end{gather*}
Furthermore, since $U_{\ell}(\D^\ell) \subseteq U_{\ell+1}(\D^{\ell+1})$,
\begin{equation}\label{eqn:Vm_subseteq_h}
     V_0\subseteq V_1\subseteq V_2\subseteq\cdots\subseteq\mathfrak h.
\end{equation}

Every strict inclusion increases the dimension by at least one. There can therefore be at most $\dim \h-r$ strict inclusions. Hence $m_0 = \min \{m \geq 0 \, : \, V_m = V_{m+1} \}$, as defined in \Cref{eqn:m0_def}, exists and
\begin{gather*}
    m_0 \leq \dim \h - r \leq d^2-r,
\end{gather*}
where the final inequality follows from $\h \leq \mathfrak \mathrm{U}(d)$ and $\dim \mathfrak \mathrm{U}(d)=d^2$.

We now prove the reverse, $V_{m_0} \supseteq \h$. 
Let $X\in V_{m_0}$, where $m_0$ is defined as previously, and $C(\phi) \in \mathcal S$.
Using Campbell's lemma (\Cref{lem:campbell}),
\begin{gather*}
\Ad_{F(\phi,t)} X =e^{\timeB\ad_{C(\phi)}}X \in V_{m_0+1} = V_{m_0}, \\ 0\leq\timeB<\varepsilon.
\end{gather*}
Since $V_{m_0}$ is a vector space,
\begin{gather*}
\frac{e^{\timeB\ad_{C(\phi)}}X-X}{\timeB}\in V_{m_0},
\\ 0<\timeB<\varepsilon.
\end{gather*}
Moreover, $V_{m_0}$ is finite-dimensional and therefore closed. Taking the limit as $\timeB\to 0^+$, which is allowed because $[0, \varepsilon) \subseteq \D_B$, yields
\begin{gather*}
[C(\phi),X]
=\lim_{\timeB\to0^+}
\frac{e^{\timeB\ad_{C(\phi)}}X-e^{0\ad_{C(\phi)}} X}{\timeB}
\in V_{m_0}.
\end{gather*}
Since $X$ and $C(\phi)$ were arbitrary, for every $C(\phi) \in \mathcal S$,
\begin{gather*}
\ad_{C(\phi)}(V_{m_0})\subseteq V_{m_0}.
\end{gather*}
Hence $\mathcal S \subseteq V_{m_0}$ contains  and, by repeated application of the inclusion $k$ times, every nested bracket
\begin{gather*}
[C(\phi_1),[C(\phi_2),\ldots,[C(\phi_{k-1}),C(\phi_k)]\ldots]] \in V_{m_0}.
\end{gather*}
These brackets, together with the generators $\mathcal S$, span the generated Lie subalgebra $\h$. It follows,
\begin{gather*}
\h \subseteq V_{m_0}.
\end{gather*}
\end{proof}

\begin{proposition}[Conditions for $\h = \g$ when $\h \subseteq \g$]\label{prop:subalgebra}
$\h$ is a Lie subalgebra of $\g$. Let $H_\h$ be the unique, connected Lie group corresponding to $\h$.
Then, $G = \overline{H_\h}$ and is connected. 
Moreover, $\h = \g \iff H_\h = \overline{H_\h}$.
\end{proposition}
\begin{proof}
For every $\phi\in\D_A$ and $\timeB\in[0,\varepsilon)$, $e^{\timeB C(\phi)}\in\Ucal_1\subseteq G.$
Since $G$ is a group by \Cref{prop:G_group}, it also contains $e^{-\timeB C(\phi)}=(e^{\timeB C(\phi)})^{-1}.$
Hence the smooth curve
\begin{equation*}
    \gamma:(-\varepsilon,\varepsilon)\longrightarrow G,
    \;
    \gamma(\timeB)=e^{\timeB C(\phi)},
\end{equation*}
is well defined and satisfies $\gamma(0)=\I$. Hence
\begin{equation}
    C(\phi)=\dot{\gamma}(0)\in T_\I G=\g.
\end{equation}
Since $\g$ is a real Lie algebra containing every $C(\phi)$, it contains $\h=\Lie_{\mathbb R}\langle C(\phi):\phi\in\D_A\rangle$.
As a Lie subalgebra, $\h$ has a unique connected, immersed Lie subgroup $H_\h \leq G$~\Cref{def:subalgebra}. 

We want to prove first that $\overline{H_\h} = \overline{\Ucal}$ using that $[0, \varepsilon) \subseteq \D_B$.
Since $C(\phi) \in \h$, the corresponding one-parameter subgroup is contained in $H_\h$:
\begin{equation*}
    e^{\timeB C(\phi)}\in H_\h
    \\
    \quad \forall \,\phi\in\D_A, \, \timeB\in\D_B.
\end{equation*}
Thus, for every 
\[
x  = ((\phi_0, \timeB_0), (\phi_1, \timeB_1), \dots, (\phi_{L-1}, \timeB_{L-1}))\in\D^L
\]
\begin{equation*}
    U_L(x)
      =\prod_{j=0}^{L-1}e^{\timeB_j C(\phi_j)}
      \in H_\h,
\end{equation*}
because $H_\h$ is a subgroup and hence closed under every finite product of its elements. Therefore
\begin{equation*}
    \forall \, L \geq 1, \quad \Ucal_L\subseteq H_\h.
\end{equation*}
Taking the union over all finite lengths $L\geq1$, and then taking closures, gives
\begin{equation*}\label{eqn:Gamma_H}
    \Ucal\subseteq H_\h, \quad \quad \overline{\Ucal}\subseteq\overline{H_\h}.
\end{equation*}

On the other hand, since $H_\h \leq G$ and $G$ is closed, 
\begin{equation*}\label{eqn:H_Gamma}
    \overline{H_\h} \subseteq G = \overline{\Ucal}.
\end{equation*}
We have therefore
\begin{equation}\label{eqn:Hh_closure}
    \overline{\Ucal} = \overline{H_\h} = G.
\end{equation}
Since $H_{\h}$ is connected and its inclusion into $\mathrm{U}(d)$ is continuous, its closure $G=\overline{H_{\h}}$ is connected.

To finish the proof, we need to show $\h = \g  \iff H_\h = \overline{H_\h}$ using connectedness and uniqueness.
Suppose $\h=\g$. Then $H_\h$ and $G$ are connected immersed Lie subgroups of $\mathrm{U}(d)$ with the same Lie algebra. By uniqueness in the Lie subalgebra--subgroup correspondence, \Cref{def:subalgebra}, we obtain
\[
    H_\h = G = \overline{H_\h}.
\]
Conversely, suppose $H_\h=\overline{H_\h}$. Then \Cref{eqn:Hh_closure} gives $H_\h=G$, and therefore
\[
    \h = \Lie(H_\h) = \Lie(G) = \g.
\]
Thus $\h=\g$ if and only if $H_\h=\overline{H_\h}$.
\end{proof}

\begin{proposition}\label{prop:h_g_submersion}
    $\h = \g$ is necessary for $U_L$ to be a submersion.
\end{proposition}
\begin{proof}
First, we prove
  \begin{gather*}
  \g = \h \iff \dim \h = \dim \g.
  \end{gather*}
The forward direction is immediate. On the other hand, \Cref{prop:subalgebra} gives $\h\subseteq\g$. Since these Lie algebras are
finite-dimensional, equality of their dimensions implies $\ h=\g$.

By \Cref{lem:im_J_subset_h}, $ \dim \h \geq \operatorname{rank}J_{L,x}$. If $U_L$ has a submersion point, then $\operatorname{rank}J_{L,x} = \dim \g$, which as we showed above implies $\h = \g$. 
\end{proof}

In the Appendix in \Cref{lem:comm_comm_sufficient}, we prove a sufficient condition to guarantee $\h = \g$: any matrix commuting with both $A$ and $B$ must be proportional to the identity. In \Cref{fig:proof} this appears on the lower right as a dashed line to indicate that it is a test conducted on $A$ and $B$ and is not required for the remainder of the overall proof.

\section{Main results}\label{sec:main_results}

Recall from \Cref{sec:proof_overview} that there are two primary technical results we wish to establish: (1) that $G$ is connected and compact, and (2) that there exists some finite $L$ such that $\mu^{*L} \ll \mu_G$.
From \Cref{prop:G_group}, we know $G$ is a compact Lie subgroup of $\mathrm{U}(d)$.
From \Cref{prop:subalgebra}, $\overline{H_\h} = G$, which implies $G$ is connected. From \Cref{prop:h_g_submersion}, $\h = \g$ is a necessary condition for a submersion point of $U_L$ to exist.
We have all the ingredients to establish the open-set irreducibility of the random walk whose $L$-step evolution is described by $U_L$, which would answer \Cref{question:topo_irred}.
The first step is to show that there exists a submerison point of $U_L$ for some finite $L$.

The coarea formula (\Cref{thm:coarea}) gives a geometric interpretation of this step. Where the differential of a map has full rank, the volume of a region in the domain can be calculated by integrating the reciprocal of the normal Jacobian over each level-set, or fiber, then integrating over the map's output values. 
The normal Jacobian measures how the map scales volume in directions perpendicular to the level sets. Taking the reciprocal of the normal Jacobian accounts for the fact that, when the map expands more in directions perpendicular to the level sets, less volume in those directions is needed to produce the same output volume.

A submersion point has a neighborhood on which the differential is surjective, so this description applies locally. The resulting output density reflects a balance between fiber size and transverse expansion: larger fibers or smaller normal Jacobians favor higher density. Uniformity requires this balance to yield a constant density relative to the Haar measure.

If we can construct a submersion point in a finite length, \Cref{prop:submersion_almost_everywhere} in the Appendix guarantees that almost every input point is a submersion point. The coarea description therefore applies to the entire pushforward measure, instead of only to the contribution from a neighborhood of our constructed point. Applying this formula then shows that the entire pushforward of input volume is absolutely continuous with respect to Haar measure, which implies open-set irreducibility.

\begin{definition}[Minimum length $L_0$ for a submersion point]
\label{def:L0}
Define
\[
    L_0:=
    \min\left\{
        L\geq1:
        \exists x\in(\D^\circ)^L
        \text{ with }\Lambda_L(x)>0
    \right\},
\]
with $L_0=\infty$ if this set is empty.
\end{definition}

\begin{lemma}[Bounding the minimum submersion length]
\label{lem:L0_bound}
Let $L_0$ be as in \Cref{def:L0}. If  $\h = \g$, $U_{L_0}$ has a submersion point in $\D^{L_0}$ with
\begin{gather*}
    \left\lceil\frac n2\right\rceil \leq L_0 \leq r+(n-r)(n-r+1) \leq d^4.
\end{gather*}
\end{lemma}
\begin{proof}
By \Cref{prop:h_g_submersion}, $\h=\g$ is necessary. For sufficiency, assume $\h=\g$.
We first construct a point, possibly with zero time coordinates, at which the extended differential has full rank. We then perturb this point into the interior of the physical parameter domain.

Recall that the tangent space $T_x\D^L$ and the Lie algebra $\g$ are real vector spaces of dimensions $2L$ and $n$, respectively. Thus, we can represent the right-trivialized differential by a real matrix:
\begin{equation}
    J_{L,x}\in\mathbb R^{n\times 2L}.
\end{equation}
Hence $J_{L,x}$ cannot have full rank unless $2L\geq n$. This gives the lower bound:
\begin{equation}
    L_0\geq\left\lceil\frac n2\right\rceil.
\end{equation}

To obtain the upper bound, we construct a submersion point. Recall \Cref{eqn:jacobian_formula}:
\begin{gather*}
J_{L,x}\left(\frac{\partial}{\partial \timeB_j}\right)
=\Ad_{\Lprod{j}}C(\phi_j),
\\
\Lprod{j}=\prod_{\ell=0}^{j-1}F(x_\ell).
\end{gather*}
Since $\dim V_0=\dim \Span_{\R}\mathcal S=r$, where $\mathcal S$ is given by \Cref{eqn:Scal}, we can choose \(\phi_0,\ldots,\phi_{r-1}\in\D_A\) such that
$\{C_{\phi_0},\ldots,C_{\phi_{r-1}}\}$ are linearly independent, and take the corresponding $\timeB_j=0 $ for $0\leq j\leq r-1$. 

We construct a point $x$ by successively extending the given initial segment
\begin{equation}
    x|_r := (x_0,\ldots,x_{r-1}).
\end{equation}
Then $\Lprod{j}=\I$ for $0 \leq j < r$, and 
\begin{gather*}
 W = \Span_\R \left \{ J_{r,x|_r}\left( \frac{\partial}{\partial \timeB_j}\right) \right \} \subset \g, \quad \dim W = r.
\end{gather*}

If $n=r$, the extended differential already has full rank at $x|_r$, and we proceed directly to the interior perturbation below. If $n>r$, we continue the construction.

Let $P=U_r(x|_r).$ By \Cref{lem:V_m_h}, there exists $m_0\leq n-r$ such that $V_{m_0}=\h,$
so $\h$ is invariant under conjugation by $P$. By hypothesis $\h = \g$, so $\Ad_P$ is a linear automorphism of $\g$. Hence
\begin{multline*}
    \g = \Ad_P(V_{m_0})
    = \Span_{\R}\Bigl\{
        \Ad_{PU_\ell(\xi)}C(\phi) : \\
        \phi\in\D_A,\quad
        0\leq\ell\leq m_0,\quad
        \xi\in(\D_A\times[0,\varepsilon))^\ell
    \Bigr\}.
\end{multline*}
As in the case $n > r$, $W$ is a proper subspace of $\g$, these spanning vectors
cannot all lie in $W$. Thus there exist
$0\leq\ell\leq m_0$,
$\xi\in(\D_A\times[0,\varepsilon))^\ell$, and
$\phi'_{0}\in\D_A$ such that
\[
    \Ad_{PU_\ell(\xi)}C(\phi'_0)\notin W.
\]
We extend the initial $r$ coordinates for our point:
\begin{equation}
    x|_{r+\ell+1}
    := (x|_r,\xi,(\phi'_0,0)).
\end{equation}
This corresponds to the unitary,
\begin{gather*}
U_{r+\ell}(x|_{r+\ell}) =U_r(x|_r)U_\ell(\xi). 
\end{gather*}
By \Cref{eqn:jacobian_formula}, the $r+\ell+1$ column of the Jacobian is
\begin{equation}
J_{r+\ell+1,x|_{r+\ell+1}}
\left(\frac{\partial}{\partial \timeB_{r+\ell}}\right)
=\Ad_{PU_\ell(\xi)}C(\phi'_0)
\notin W.
\end{equation}

Recalling that the $\Rprod{j}$ factors do not appear in \Cref{eqn:jacobian_formula}, we can append $\xi$ to the total tuple leaving the previous $\frac{\partial}{\partial \timeB}$ columns of the Jacobian unchanged. 
Therefore this process adds a new $\frac{\partial}{\partial \timeB}$ column outside $W$. So the rank increases by at least one, while the product length increases by $\ell+1\leq m_0+1$.

We can repeat this process until we have a full rank Jacobian. At step $k$, define
\begin{align}
    r_{k+1}&:=r_k+\ell_k+1,
    \\
    x|_{r_{k+1}}
    &:= (x|_{r_k},\xi_k,(\phi'_k,0)).
\end{align}
After $c$ such extensions, set $x:=x|_{r_c}$.

Since $V_{m_0}=\g$ and $\Ad_P$ is a linear automorphism of $\g$, the available new $\frac{\partial}{\partial \timeB}$ columns beyond the initial $\{ C(\phi_j) \}$ span $\Ad_P(V_{m_0})=\g$. Thus, whenever the existing columns
span a proper subspace of $\g$, a further linearly independent column can be added by the method described above. Since $\dim\g=n$, after $c\leq n-r$ extensions the $\frac{\partial}{\partial \timeB}$ columns span $\g$.

In either case, the construction gives a length
\[
    K\leq r+(n-r)(m_0+1)
\]
and a point
\[
    x=((\phi_j,t_j))_{j=0}^{K-1}
    \in(\D_A\times[0,\varepsilon))^K
\]
such that $\Lambda_K(x)>0$.
Some of the times $t_j$ may equal zero.

For $0<\eta<\varepsilon$, define
\begin{equation}
    t_j^{(\eta)}
    :=
    \begin{cases}
        \eta, & t_j=0,\\
        t_j, & t_j>0,
    \end{cases}
    \quad
    x^{(\eta)}
    :=((\phi_j,t_j^{(\eta)}))_{j=0}^{K-1}.
\end{equation}
Then $x^{(\eta)}\in(\D^\circ)^K$ and $x^{(\eta)}\to x$ as $\eta\to0^+$. Since $\Lambda_K$ is continuous and $\Lambda_K(x)>0$, we have
\begin{equation}
    \Lambda_K(x^{(\eta)})>0
\end{equation}
for all sufficiently small $\eta>0$.
Thus $U_K$ has an interior submersion point using strictly positive evolution times. The perturbation does not change the product length, so $L_0\leq K$.

The upper bound on the length for a submersion point follows:
\begin{gather}
L_0 \leq r+(n-r)(m_0+1) \leq r+(n-r)(n-r+1).
\end{gather}
Since \(1\leq r\leq n\leq d^2\),
\begin{align*}
r+(n-r)(n-r+1) &= n+(n-r)^2 \\ &\leq n+(n-1)^2 \leq n^2 \leq d^4.
\end{align*}
\end{proof}

Since we have constructed a submersion point in finite length, \Cref{prop:submersion_almost_everywhere} guarantees that almost every input point in $\D^{L_0}$ is a submersion point with respect to the normalized Lebesgue measure $\lambda^{\otimes L_0}$. Hence, the earlier coarea discussion implies that the pushforward measure is absolutely continuous with respect to Haar measure. This is formalized in \Cref{thm:non0_ac}, which brings together all previous results as shown in \Cref{fig:proof} in the dark purple box with wavy lines.
Since $G$ is connected, absolute continuity of $\mu^{*L_0}$ with respect to Haar measure implies open-set irreducibility and rules out periodic obstructions; we formalize this statement in \Cref{cor:os_irred}.

Define the set of submersion points:
\begin{equation}\label{eqn:N_L}
    \mathcal N_L
    :=
    \{x\in(\D^\circ)^L:\Lambda_L(x)>0\}.
\end{equation}
For $L$ satisfying $2L\geq n$ and $y\in G$, define the fiber volume (of the regular part of the fiber) by integrating with respect to the measure $d \Hd^{2L-n}$ over the fiber over $y$, i.e. the level set of $\{ x \, : \, U_L(x)=y\}$:
\begin{equation}\label{eqn:fiber_vol}
\begin{split}
\fibvol_L(y) &= \Hd^{2L-n}\bigl(U_L^{-1}(y)\cap \mathcal{N}_L\bigr) \\ &=  \int_{U_L^{-1}(y)\cap \mathcal{N}_L} d\Hd^{2L-n}(x).
\end{split}
\end{equation}
Equip $\D^L$ with the Euclidean metric inherited from $\R^{2L}$, and $G$ with the real Hilbert-Schmidt inner product as the chosen bi-invariant Riemannian metric. This allows us to write the total volumes:
\begin{equation}
\Vol=\vol(\mathcal D^{L_0}),
\;
\Vol_G=\vol(G).
\end{equation}
Since $\lambda$ is the normalized Lebesgue measure on $\D$, the product measure obeys
\begin{equation}
    d\lambda^{\otimes L_0}(x) =\frac1{\Vol}\,d\vol_{\mathcal D^{L_0}}(x).
\end{equation}
Similarly, as $\mu_G$ is the normalized Haar measure on $G$,
\begin{equation}
d\mu_G(y)=\frac1{\Vol_G}\,d\vol_G(y).
\end{equation}

\begin{theorem}[Absolute continuity of the pushforward measure almost everywhere.]
\label{thm:non0_ac}
Let $L_0$ be as defined and bounded in \Cref{lem:L0_bound}. Let $\mu=F_{\#}(\lambda)$ be the pushforward of the normalized Lebesgue measure $\lambda$ on $\D$. Then, if $\h=\g$, 
\begin{gather*}
\mu^{*L_0} =(U_{L_0})_{\#}(\lambda^{\otimes L_0}) \ll \mu_G.
\end{gather*}
\end{theorem}
\begin{proof}
By \Cref{lem:L0_bound}, $U_{L_0}$ has a submersion point if  $\h = \g$. Then, \Cref{prop:submersion_almost_everywhere} implies
\begin{gather*}
\lambda^{\otimes L_0}(\D^{L_0}\setminus \mathcal{N}_{L_0})=0.
\end{gather*}

From \Cref{eqn:normal_jacobian}, $\sqrt{\Lambda_{L_0}}$  is the normal Jacobian of $U_{L_0}$, computed with respect to the Riemannian metrics $d\vol_{\mathcal D^{L_0}}$ and $d\vol_G$.

Because $U_L$ is smooth, its derivative is bounded on a sufficiently small neighborhood around each point of $\D^L$. The mean-value inequality therefore bounds the change in the output by a constant times the change in the input, so $U_L$ is locally Lipschitz~\cite[Thm.~9.19]{rudin1976}.

We apply the coarea formula on the open manifold $(\D^\circ)^{L_0}$. Its complement in $\D^{L_0}$ has $\lambda^{\otimes L_0}$-measure zero, so this restriction does not change the pushforward measure.
The set $\mathcal N_{L_0}$ lies in this interior by definition, and its complement there also has measure zero. Thus all fiber integrals below are taken over the regular interior fibers, with the same normalization $\Vol=\vol(\D^{L_0})$.

Since $\lambda^{\otimes L_0}(\mathcal D^{L_0}\setminus \mathcal N_{L_0})=0$, restricting the following integral to $\mathcal N_{L_0}$ does not change the pushforward measure.
For every Borel set $E \in \mathcal B (G)$, 
\begin{widetext}
\begin{align}
\mu^{*L_0}(E) =\lambda^{\otimes L_0}\bigl(U_{L_0}^{-1}(E)\cap \mathcal{N}_{L_0} \bigr) \nonumber &=\frac1{\Vol} \int_{U_{L_0}^{-1}(E)\cap \mathcal{N}_{L_0}} d\vol_{\mathcal D^{L_0}}(x)\nonumber \\
\text{(by \Cref{thm:coarea}, since $F$ locally Lipschitz)} \;&=\frac1{\Vol} \int_E \left( \int_{U_{L_0}^{-1}(y)\cap \mathcal{N}_{L_0}} \frac1{\sqrt{\Lambda_{L_0}(x)}} \nonumber  \,d\Hd^{2L_0-n}(x) \right) d\vol_G(y)\nonumber \\
&=\int_E \underbrace{\frac{\Vol_G}{\Vol}
\int_{U_{L_0}^{-1}(y)\cap \mathcal{N}_{L_0}}
\frac1{\sqrt{\Lambda_{L_0}}(x)}
\,d\Hd^{2L_0-n}(x)}_{=:\, \rho(y)} \nonumber \,d\mu_G(y) \nonumber\\
&= \int_E \rho(y) \,d\mu_G(y) \label{eqn:mu_rho}
\end{align}
\end{widetext}
$\Hd^{2L_0-n}$ is the induced measure on the fibers, or level sets, $U_{L_0}^{-1}(y) = \{x \in \D^{L_0} \, : \, U_{L_0}(x) = y\}$. 

Taking $E=G$ in \Cref{eqn:mu_rho} gives
\[
    \int_G \rho(y)\,d\mu_G(y)
    =\mu^{*L_0}(G)=1.
\]
because $\mu^{*L_0}$ is a probability measure supported on $G$.
Thus $\rho$ is non-negative and measurable.
Then, for every $E\in\mathcal B(G)$,
\[
    \mu_G(E)=0
    \quad\Longrightarrow\quad
    \mu^{*L_0}(E)
    =\int_E\rho(y)\,d\mu_G(y)=0,
\]
Therefore $\mu^{*L_0}\ll\mu_G$.
\end{proof}

\begin{corollary}[Open-set irreducibility]\label{cor:os_irred}
Suppose $G$ is compact and connected and $\mu^{*L_0}\ll\mu_G$. Then, for every starting point $g\in G$ and every nonempty open set $O\subseteq G$, there exists $k\geq1$ such that
\[
    \mu^{*(kL_0)}(Og^{-1})>0.
\]
\end{corollary}

\begin{proof}
Let $Y_1,Y_2,\ldots$ be independent random elements with common distribution $\mu$, and define the random walk by
\[
    X_0=g,\qquad X_L=Y_LX_{L-1}=Y_L\cdots Y_1g.
\]
Since $Y_L\cdots Y_1$ has distribution $\mu^{*L}$,
\[
    \mathbb P_g(X_L\in O)=\mu^{*L}(Og^{-1}).
\]

By Anoussis--Gatzouras, Remark~(5) following Corollary~4.2~\cite{anoussis2004},
connectedness of $G$ and absolute continuity of $\mu^{*L_0}$ imply that its support is not contained in any proper closed subgroup of $G$.

The closure of the semigroup generated by $\operatorname{supp}\mu^{*L_0}$ is a subgroup, by compactness (every sequence has a convergent subsequence whose limit remains in the set, following the same style of argument we made in \Cref{prop:G_group}), and therefore equals $G$. Thus the nonempty open set $Og^{-1}$ intersects the support of some convolution power $\mu^{*kL_0}$, giving
\[
    \mathbb P_g(X_{kL_0}\in O)
    =\mu^{*kL_0}(Og^{-1})>0.
\]
Hence the random walk is open-set irreducible by \Cref{def:topo_irred}.
\end{proof}

Now, we aim show that the distribution of endpoints of this random walk converge in total variation to the Haar measure on $G$.
From \Cref{thm:non0_ac}, as \Cref{fig:proof} illustrates we can apply the results of Bhattacharya~\cite[Thm.~3]{bhattacharya1972} and Anoussis--Gatzouras~\cite[Cor.~4.2]{anoussis2004} to convert the absolute continuity of $\mu^{L_0}$ into convergence to the Haar measure as a function of  further convolutions.
The intuition behind these results is that absolute continuity of $\mu^{*L_0}$ ensures that the random walk has spread out over a set of positive Haar measure (open-set irreducibility). Since $G$ is connected, subgroup and periodic obstructions to mixing are avoided. Further convolutions average translates of the distribution, iteratively decreasing nonuniformity and driving it toward Haar measure.

\begin{theorem}[Exponential convergence to Haar with uniform sampling on $\D$]\label{thm:exp_conv}
Assume the minimum product length $L_0$ at which $U_{L_0}$ is a submersion is finite and bounded as in \Cref{lem:L0_bound}.
If  $\h = \g$,\footnote{Here $\operatorname{spr}(\widehat{\mu^{*L_0}}(R))$ denotes the spectral radius of the Fourier coefficient matrix
\[
\widehat{\mu^{*L_0}}(R)=\int_G R(g^{-1})\,d\mu^{*L_0}(g).
\] See Appendix \Cref{apdx:rep_theory} for an introduction to the  relevant material of representation theory.} 
\begin{align*}
\lim_{L'\to\infty}
&\|\mu^{*L'}-\mu_G\|_{\mathrm{TV}}^{1/L'}\\
&=
\left[
\max_{R\in\Grep\setminus\{\mathbf1_G\}}
\operatorname{spr}\bigl(\widehat{\mu^{*L_0}}(R)\bigr)
\right]^{1/L_0} < 1.
\end{align*}
\end{theorem}

\begin{proof}
By \Cref{thm:non0_ac}, the pushforward measure $\mu^{*L}$ has a non-zero absolutely continuous component with respect to $\mu_G$ for $L \geq L_0$ if $\h = \g$, which is given by hypothesis. Moreover, $G$ is compact by \Cref{prop:G_group} and connected by \Cref{prop:subalgebra}.

$\mu^{*L}$ on $G$ is a satisfactory input to Bhattacharya's Thm.~3~\cite{bhattacharya1972}\footnote{Restated in this work as \Cref{thm:bhattacharya}}, which gives the claimed exponential convergence of $\mu^{*L}$ to $\mu_G$, for all $L \geq L_0$.

We can give a quantitative description of the asymptotic convergence rate using representation theory. 
\[
\gamma := \max_{R \in \Grep \setminus \{\mathbf{1}_G\}} \operatorname{spr}\bigl(\widehat{\mu^{*L_0}}(R)\bigr),
\]
where $\Grep$ is the unitary dual of $G$, consisting of the unitary equivalence classes of continuous irreducible unitary representations of $G$.

Since $G$ is compact and connected and $\mu^{*L_0} \ll \mu_G$, $(G,\mu^{*L_0})$ is adapted and $\mu^{*L_0}$ is strictly aperiodic. Thus the hypotheses of Anoussis-Gatzouras~\cite[Cor.~4.2]{anoussis2004}\footnote{Restated in \Cref{cor:anoussis}.} hold, and
\[
  \lim_{L' \to \infty} \|\mu^{*L'L_0} - \mu_G\|^{1/{L'}}
  = \gamma < 1.
\]
Therefore $\gamma$ is the asymptotic exponential decay factor per block of $L_0$ steps. 
To go from block lengths to arbitrary lengths, set $a_m:=\|\mu^{*m}-\mu_G\|_{\mathrm{TV}}$.
Convolution by a probability measure contracts total variation and leaves Haar measure invariant, so $a_m$ is nonincreasing. Writing $m=qL_0+s$, with $0\leq s<L_0$,
gives
\[
    a_{(q+1)L_0}\leq a_m\leq a_{qL_0}.
\]
Taking $m$th roots gives
\[
    a_{(q+1)L_0}^{1/m}
    \leq a_m^{1/m}
    \leq a_{qL_0}^{1/m}.
\]
Rewrite the upper bound as
\[
    a_{qL_0}^{1/m}
    =\left(a_{qL_0}^{1/q}\right)^{q/m}.
\]
The expression inside parentheses tends to $\gamma$,
while $q/m\to1/L_0$, so this bound tends to
$\gamma^{1/L_0}$. Similarly, the lower bound satisfies
\[
    a_{(q+1)L_0}^{1/m}
    =\left(a_{(q+1)L_0}^{1/(q+1)}\right)^{(q+1)/m}
    \longrightarrow\gamma^{1/L_0}.
\]
Since $a_m^{1/m}$ lies between two quantities with
the same limit, it follows that
\[
    \lim_{m\to\infty}a_m^{1/m}=\gamma^{1/L_0}.
\]

\end{proof}

\begin{remark}
Since the representation $R(g) = g$ distinguishes every pair of distinct elements of $G$, it is called faithful. For a compact group, tensor products of a faithful representation and its complex conjugate contain every irreducible representation as an invariant component; hence writing $R_{a,b}(g)=g^{\otimes a}\otimes\overline g^{\otimes b}$, over all $a,b\geq0$, recovers every irrep up to equivalence \cite[Section~5.d.]{bekka2008}.
To obtain the irreps numerically, we can use RepLAB~\cite{rosset2021} to decompose each chosen $R_{a,b}$ into irreducible pieces.
\end{remark}

With this remark, we have an approach that could be implemented in principle on a case-by-case basis for a specific $A$ and $B$; we leave explicit examples with specific experimental considerations to follow-up work.

We can also obtain an expression for a non-asymptotic rate.
\begin{remark}\label{rem:doeblin}
A Doeblin minorization gives a non-asymptotic, but possibly more conservative bound~\cite[Thm.~8]{roberts2004}. If for all $E\in\mathcal B(G)$,
\[
    \mu^{*K}(E)\geq \alpha\,\mu_G(E)
\]
for some $K\geq1$ and constant $0<\alpha<1$, then
\begin{equation}
    \|\mu^{*L}-\mu_G\|_{\mathrm{TV}}
    \leq 2(1-\alpha)^{\lfloor L/K\rfloor}
    \; \;\forall L\geq0,
\end{equation}
where $\|\nu-\eta\|_{\mathrm{TV}}
=2\sup_{E\in\mathcal B(G)}|\nu(E)-\eta(E)|$.\footnote{See \Cref{def:total_variation}.}
\end{remark}
We can find such an $\alpha$ using the coarea formula, lower bounding $\rho(y)$ in \Cref{eqn:mu_rho} almost everywhere with respect to $\mu_G$. The Jacobian is easy to upper bound, but lowering bounding the fiber volume requires more careful consideration.
We leave this, for specific $A$ and $B$, to future work.

\begin{theorem}[Exact Haar sampling by non-uniform measure
on $\mathcal D^L$]
\label{thm:exact}
Assume $\h=\g$. Let $L_0$ is as defined in \Cref{lem:L0_bound}. Suppose that for almost every $y \in G$ with respect to $\mu_G$, \footnote{This implies $\rho(y)>0$ for almost every $y \in G$ with respect to $\mu_G$, which is called positivity.}
\[
    0<\fibvol_{L_0}(y)<\infty\,
\]
and define\footnote{Here the essential supremum is taken with respect to $\lambda^{\otimes {L_0}}$ (the normalized Lebesgue measure on $\D^{L_0}$).}
\[
    0<M:=
    \operatorname*{ess\,sup}_{x\in\mathcal D^{L_0}}
    \frac{\sqrt{\Lambda_{L_0}(x)}}
         {\fibvol_{L_0}(U_{L_0}(x))}
    <\infty.
\]
The ratio inside the essential supremum is defined to be zero wherever its denominator is zero or infinite.

Repeatedly sample $ X \sim \lambda^{\otimes L_0}$, accepting it with probability
\begin{gather*}
p(X)= \frac1M \frac{\sqrt{\Lambda_{L_0}(X)}}{ \fibvol_{L_0}(U_{L_0}(X))}
\end{gather*}
Then the output $\{ U_{L_0}(X) \}$, conditioned on acceptance, has distribution $\mu_G$.
\end{theorem}

\begin{proof}
\Cref{thm:non0_ac} guarantees $\mu^{*L} \ll \mu_G$ for $L \geq L_0$, by the hypothesis $\h = \g$.

Draw a random $X \sim \lambda^{\otimes L_0}$ and let $p: \D^{L_0} \to [0,1]$ be an unspecified measurable acceptance probability function of $X$.
Independently draw a number $T\sim\operatorname{Unif}([0,1])$, and define the acceptance event
\begin{gather*}
\mathrm{Acc} := \{T\le p(X)\}.
\end{gather*}
We want to determine $p$ such that $\mathbb P(\mathrm{Acc})>0$ and the conditional distribution of $U_{L_0}(X)$ given acceptance yields the Haar measure on $G$. 
That is, for $E \in \mathcal{B}(G)$ apply the definition of conditional probability:
\begin{align*}
\mathbb{P} \left ( \{ U_{L_0}(X) \in E \} \, | \, \mathrm{Acc}  \right ) &= \frac{\mathbb{P} \left ( \{ U_{L_0}(X) \in E \} \cap \mathrm{Acc}  \right )}{\mathbb{P}(\mathrm{Acc})} \\ &= \mu_G(E).
\end{align*}
The numerator can be written as:
\begin{align*}
    \mathbb{P} \left ( \{ U_{L_0}(x) \in E \} \cap \mathrm{Acc}  \right ) &= \int_{U_{L_0}^{-1}(E)} p(x) \, d\lambda^{\otimes {L_0}}(x)\\
    &= \frac1{\Vol}\int_{U_{L_0}^{-1}(E)} p(x)\, d\vol_{\D^{L_0}}(x).
\end{align*}
We can now use the coarea formula to evaluate it, for $y = U_{L_0}(x)$:
\begin{widetext}
\begin{align}
    &\mathbb{P} \left ( \{ U_{L_0}(x) \in E \} \cap \mathrm{Acc}  \right ) \nonumber \\
    &= \frac1{\Vol} \int_E \Bigl[ \int_{U_{L_0}^{-1}(y) \cap \mathcal N_{L_0}} p(x) \frac1{\sqrt{\Lambda_{L_0}(x)}} \, d\Hd^{2L_0 - n}(x)  \Bigr]  \, d \vol_G(y) \nonumber \\
    &= \frac{\Vol_G}{\Vol} \int_E \Bigl[ \int_{U_{L_0}^{-1}(y) \cap \mathcal N_{L_0}} p(x) \frac1{\sqrt{\Lambda_{L_0}(x)}} \, d\Hd^{2L_0 - n}(x)  \Bigr] \, d \mu_G(y) \label{eqn:coarea_prob}
\end{align}
\end{widetext}
We need the inner integral to evaluate to a constant independent of $y$.

It is convenient to take $p(x) \propto \sqrt{\Lambda_{L_0}(x)}$ in order to cancel the factor of the normal Jacobian in the denominator. Recall from \Cref{eqn:fiber_vol}
\begin{gather*}
\fibvol_{L_0}(y) = \int_{U_{L_0}^{-1}(y) \cap \mathcal N_{L_0}} \, d \Hd^{2{L_0} - n}(x), 
\end{gather*}
so we should also take $p(x) \propto \frac1{\fibvol_{L_0}(U_{L_0}(x))}$. In order to be a valid probability, $ 0\leq p(x) \leq 1$, so we take $p(x) \propto \frac1M$, where 
\begin{gather*}
M = \esssup\limits_{x \in \D^{L_0}} \frac{\sqrt{\Lambda_{L_0}(x)}}{\fibvol_{L_0}(U_{L_0}(x))}.
\end{gather*}
By hypothesis, $M<\infty$.
Hence we claim:
\begin{equation}
    p(x) = \frac{\sqrt{\Lambda_{L_0}(x)}}{M \fibvol_{L_0}(U_{L_0}(x))}.
\end{equation}
On any exceptional null set where this formula does not give a value in $[0,1]$, set $p(x)=0$.
For $\mu_G$-almost every $y$, the inner integral in the coarea formula equals
\begin{align*}
    \int_{U_{L_0}^{-1}(y) \cap \mathcal N_{L_0}} p(x) \frac1{\sqrt{\Lambda_{L_0}(x)}} \, d\Hd^{2L_0 - n}(x) = \frac1M.
\end{align*}
We have shown, therefore, by \Cref{eqn:coarea_prob}:
\begin{align*}
\mathbb{P} \left ( \{ U_{L_0}(X) \in E \} \cap \mathrm{Acc}  \right ) &= \frac{\Vol_G}{M\Vol} \int_E  \, d \mu_G \\ &= \frac{\Vol_G}{M\Vol} \, \mu_G(E). 
\end{align*}
Take $E = G$, since $U_{L_0}(X) \in G$ always for $X \in \D^{L_0}$:
\begin{align*}
\mathbb P(\mathrm{Acc}) = \mathbb{P} \left ( \{ U_{L_0}(X) \in G \} \cap \mathrm{Acc}  \right ) &= \frac{\Vol_G}{M\Vol} \int_G  \, d \mu_G \\ &= \frac{\Vol_G}{M\Vol} > 0
\end{align*}
Hence we can apply the definition of conditional probability,
\begin{gather*}
\mathbb{P} \left ( \{ U_{L_0}(X) \in E \} \, | \, \mathrm{Acc}  \right ) = \frac{\frac{\Vol_G}{M\Vol} \, \mu_G(E) }{\frac{\Vol_G}{M\Vol}} = \mu_G(E),
\end{gather*}
as desired.
\end{proof}

The following immediate corollary connects our convergence results directly to the construction of unitary designs.

\begin{corollary}[Unitary designs.]\label{cor:designs}
    If $G=\mathrm{U}(d)$, the alternating conjugated evolution random walk is open-set irreducible by \Cref{thm:non0_ac}, and its endpoint distribution converges exponentially in total variation to Haar measure on $\mathrm{U}(d)$ by \Cref{thm:exp_conv}. Consequently, for every fixed $k$, its $k$th moments converge to the Haar moments, yielding approximate unitary $k$-designs. Moreover, under the additional fiber-volume and boundedness hypotheses of \Cref{thm:exact}, its sampling construction realizes Haar measure exactly at length $L_0$, yielding an exact unitary $k$-design for every $k$.
\end{corollary}

\section{Discussion}
We began with a question about whether a Hamiltonian with a global control operation could induce a random unitary walk that has (1) open-set irreducibility and (2) convergence to the Haar measure---i.e., that the walk both can reach every neighborhood after some finite length on the compact group that is generated, and that the probability assigned to each neighborhood converges to the normalized uniform Haar measure.

We introduced a mechanism of products of evolutions under a conjugated Hamiltonian. As discussed in Appendix \Cref{apdx:examples}, this mechanism applies to a broad range of systems, including  neutral atoms, bosons/fermions in optical lattices, superconducting circuits, and trapped ions.
With this contruction, \Cref{thm:non0_ac} establishes absolute continuity with respect to the Haar measure after finitely many steps. On the compact connected group considered here, this is stronger than (1) and allows us
to establish (2), with exponential convergence in total variation, in \Cref{thm:exp_conv}.
Central to achieving these results is the construction of a  submersion point of $U_L$ in finite length in \Cref{lem:L0_bound}, which we are able to show analytically because of the structure of conjugated evolution operator.

A key point of our results is that the required walk length is controlled by the dimension of the Lie algebra generated by the control operation and Hamiltonian, rather than directly by the size of the overall Hilbert space. The necessary lower bound has a simple geometric origin: the Jacobian must have enough columns to span that algebra. With too few input parameters, it cannot have full rank.
We also give a simple sufficient condition in the appendix for guaranteeing the closure of the subgroup generated by the conjugated base Hamiltonian, \Cref{lem:comm_comm_sufficient}, for constructing a submersion point: that a matrix commuting with both $A$ and $B$ must be proportional to $\I$. This provides a practical starting point for checking whether the argument applies to a given Hamiltonian and control. 

These results establish both the possibility of generating uniform randomness through global control and quantitative guarantees for its convergence. Turning these guarantees into practical experimental protocols remains an open direction. Our bounds control convergence of the full distribution in total variation, whereas many applications require only agreement with the Haar measure in finitely many moments, as captured by unitary designs. This leaves room for potentially sharper, application-specific bounds. Evaluating the convergence rates for concrete Hamiltonians, optimizing the control sequences, and incorporating finite coherence times and various control errors are natural next steps toward making the construction experimentally implementable.

\begin{acknowledgments}
The authors thank Andrew Daley, B\'alint Koczor, and Dorian Gangloff for helpful discussions.
O.S. thanks Zhe Xian Koong for substantive feedback on \Cref{fig:overall}. O.S. thanks  Zhe Xian Koong and Leanne Reeve for helpful comments on \Cref{apdx:examples}.
O.S. acknowledges support from a Doctoral Landscape DPhil award.
S. S. acknowledges support from the Royal Society through a University Research Fellowship.
A.S.M.-R. was supported by the EPSRC UK Multidisciplinary Centre for Neuromorphic Computing (grant UKRI982).\\
\end{acknowledgments}

\paragraph*{AI usage statement.}
We used Anthropic Claude Sonnet 4.6, accessed through Claude.ai in July 2026, to explore some  proof strategies for \Cref{thm:exp_conv} and to discover related literature in \Cref{sec:background} and locate textbook references for the definitions in Appendix \Cref{apdx:def}. Additionally we used ChatGPT Sol 5.6 and Astra 6 in August - September 2026 to improve the clarity of the manuscript throughout, in particular to make discussion of the application to experimental Hamiltonians in Appendix \Cref{apdx:examples} more pedagogical.
We iterated with Astra 6 in September 2026 to draft the tikz version of \Cref{fig:proof} from our hand-drawn design.

\bibliography{squash.bib}

\pagebreak 
\clearpage
\newpage

\appendix

\section{Definitions}\label{apdx:def}

\subsection{Topology}
\begin{definition}[Topological space {\rm\cite[Ch~2]{munkres2000}}]
\label{def:topspace}
A \emph{topology} on a set $X$ is a collection $\tau$ of subsets of $X$, whose members are called \emph{open sets}, such that
\begin{enumerate}
  \item $\emptyset \in \tau$ and $X \in \tau$;
  \item $\tau$ is closed under arbitrary unions;
  \item $\tau$ is closed under finite intersections.
\end{enumerate}
The pair $(X,\tau)$ is a \emph{topological space}; we write $X$ when $\tau$ is understood. A set $C \subseteq X$ is \emph{closed} if $X \setminus C$ is open. A \emph{neighborhood} of $x \in X$ is any open set containing $x$.
\end{definition}

\begin{definition}[Hausdorff and second countable {\rm\cite[Ch.~2]{munkres2000}}]
\label{def:hausdorff}
A topological space $X$ is
\begin{itemize}
  \item \emph{Hausdorff} if any two distinct points can be separated by disjoint neighborhoods: for all $x \neq y$ in $X$ there are open sets
  $U \ni x$ and $V \ni y$ with $U \cap V = \emptyset$;
  \item \emph{second countable} if the topology admits a countable \emph{base}, i.e.\ a countable family $\mathcal B \subseteq \tau$ such that every open set is a union of members of $\mathcal B$.
\end{itemize}
\end{definition}

\begin{definition}[Compactness {\rm\cite[Ch.~2, Ch.~3, Thm.~26.5]{munkres2000}}]\label{def:compact}
A \emph{cover} of a topological space $X$ is a collection of subsets whose union is $X$.
An \emph{open cover} is a cover where every individual set in the collection is an open set. A \emph{subcover} is a smaller subset of a cover whose union is still $X$. A \emph{finite subcover} is a subcover of finite length.
$X$ is \emph{compact} if every open cover has a finite subcover.

In a \emph{Hausdorff space},  every pair of distinct points in the topological space $X$ has disjoint neighborhoods. Hence compactness implies closure for Hausdorff spaces.
For subsets of $\R^n$ (hence of $M_d(\C)$) this is equivalent to being closed and bounded (the Heine-Borel theorem). Continuous images of compact spaces are compact.
\end{definition}

\begin{definition}[Connectedness {\rm\cite[Ch.~3; Thm~23.3, 23.4, 23.5]{munkres2000}}]\label{def:connected}
 A topological space $X$ is \emph{connected} iff the only subsets of $X$ that are both open and closed in $X$ are $\emptyset$ and the whole of $X$; that is, $X$ cannot be written as the union of two disjoint nonempty
open subsets. The image of a connected space under a continuous map is connected; if $A \subset X$ is connected and $ A \subset B \subset \overline A$, then $B$ is also connected.
\end{definition}

\begin{definition}[Borel sigma algebra]\label{def:sigma_algebra}
Let \((X,\tau)\) be a topological space. The Borel $\sigma$-algebra on \(X\), denoted by \(\mathcal{B}(X)\), is the smallest $\sigma$-algebra containing every open subset of \(X\); that is,
\begin{gather*}
\mathcal{B}(X)=\sigma(\tau).
\end{gather*}
Its elements are called Borel sets. Thus, the collection of Borel sets contains all open subsets of \(X\) and is closed under complements relative to \(X\) and countable unions~\cite[Sections~1.2 and~1.3]{tassion2025}.
\end{definition}

\subsection{Measure theory}

\begin{definition}[Haar measure on $H$]\label{def:haar_G}
Let $H \leq \mathrm{U}(d)$ be a compact Lie group.
Let $\mathcal B(H)$ denote the Borel $\sigma$-algebra of $H$.\footnote{See Def.\ref{def:sigma_algebra}} We denote the normalized Haar measure on $H$ by
\begin{gather*}
    \mu_H:\mathcal B(H) \longrightarrow [0,1],
\end{gather*}
which is the unique regular Borel probability measure satisfying
\begin{gather*}
    \mu_H(hE)=\mu_H(Eh)=\mu_H(E) \quad \forall h\in H, \, E\in\mathcal B(H)
\end{gather*}
and $\mu_H(H)=1$~\cite[Corollary~9.3.2.]{cohn1980}.
\end{definition}

\begin{definition}[Absolute continuity of measures {\rm~\cite[Ch.~4.2]{cohn1980}}]\label{def:a.c.}
Let $(X,\mathcal F)$ be a measurable space where $X$ is a set and $\mathcal F$ is a Borel $\sigma$-algebra of measurable subsets of $X$.\footnote{See \Cref{def:sigma_algebra}.}

Let $\mu$ and $\eta$ be measures on this space. $\mu$ is absolutely continuous with respect to $\eta$, written $\mu\ll\eta,$ if for every $E \in \mathcal F$,
\begin{gather*}
\eta(E)=0 \quad \Longrightarrow \quad \mu(E)=0
\end{gather*}
\end{definition}

\begin{definition}[Total-variation norm and distance]
\label{def:total_variation}
Let $\nu$ and $\eta$ be probability measures on a measurable
space $(X,\mathcal F)$. Define
\[
    \|\nu-\eta\|_{\mathrm{TV}}
    :=2\sup_{E\in\mathcal F}|\nu(E)-\eta(E)|
\]
and
\[
    d_{\mathrm{TV}}(\nu,\eta)
    :=\sup_{E\in\mathcal F}|\nu(E)-\eta(E)|
    =\frac12\|\nu-\eta\|_{\mathrm{TV}}.
\]
Thus the norm takes values in $[0,2]$, while the distance
takes values in $[0,1]$.

If $\nu$ and $\eta$ have densities $f$ and $g$ with respect
to a common dominating measure $\lambda$, then
\[
    \|\nu-\eta\|_{\mathrm{TV}}
    =\int_X|f-g|\,d\lambda,
    \quad
    d_{\mathrm{TV}}(\nu,\eta)
    =\frac12\int_X|f-g|\,d\lambda.
\]
\end{definition}

\begin{corollary}[Total variation distance and approximate $k$-designs]\label{cor:TV_tdesign}
Let $G\subseteq \mathrm \mathrm{U}(d)$ be a compact subgroup,
let $\mu$ be a Borel probability measure on $G$, and let $\mu_G$
be the normalized Haar measure on $G$. Set
\[
    \delta
    :=d_{\mathrm{TV}}(\mu,\mu_G)
    =\frac12\|\mu-\mu_G\|_{\mathrm{TV}}.
\]
Then, for every integer $k\geq 1$,
\begin{gather*}
    \left\|
        \mathscr M_\mu^{(k)}-\mathscr M_{\mu_G}^{(k)}
    \right\|_\diamond
    \leq 2\delta.
\end{gather*}
Thus $\mu$ is a $2\delta$-approximate $k$-design relative to
Haar measure on $G$. For $G=\mathrm \mathrm{U}(d)$, this is the standard diamond norm definition of an approximate unitary $k$-design given in~\cite[CH.~7, Def.~41]{mele2024}.
\end{corollary}

\begin{proof}
For $U\in G$, define the unitary channel
\begin{gather*}
    \mathscr V_{U,k}(O)
    := U^{\otimes k}O(U^\dagger)^{\otimes k},
    \\
    O\in\mathcal L\bigl((\mathbb C^d)^{\otimes k}\bigr),
\end{gather*}
where $\mathcal L(\mathcal H)$ denotes the space of all linear operators on the Hilbert space $\mathcal H$.
Write the corresponding moment channel as
\begin{gather*}
    \mathscr M_\mu^{(k)}
    := \int_G \mathscr V_{U,k}\,d\mu(U),
\end{gather*}
extending the Haar moment operator of
\cite[Definition~4, Section~3]{mele2024} to $\mu$.
Since $\|\mathscr V_{U,k}\|_\diamond=1$ for every $U\in G$,
the triangle inequality for integration with respect to a signed measure gives
\begin{align*}
    \left\|
        \mathscr M_\mu^{(k)}-\mathscr M_{\mu_G}^{(k)}
    \right\|_\diamond
    &=
    \left\|
        \int_G \mathscr V_{U,k}\,d(\mu-\mu_G)(U)
    \right\|_\diamond \\
    &\leq
    \int_G \|\mathscr V_{U,k}\|_\diamond\,
        d|\mu-\mu_G|(U) \\
        &=|\mu-\mu_G|(G)\\
    &=\|\mu-\mu_G\|_{\mathrm{TV}}\\
    &=2d_{\mathrm{TV}}(\mu,\mu_G)\\
    &=2\delta.
\end{align*}
\end{proof}

\begin{definition}[Open-set irreducibility{\rm~\cite[Sec.~6.1.2]{meyn2009}}]
\label{def:topo_irred}
Let $G$ be a compact Hausdorff group, let $\mu$ be a regular Borel probability measure on $G$, and define the random walk
\[
X_0=x,\qquad X_n= Y_n \cdots Y_1 x,
\]
where $Y_1,Y_2,\ldots$ are independent random variables with common law $\mu$, i.e. $Y_i$ are independent and identically distributed (i.i.d.) with distribution $\mu$. Its transition kernel and its $n$-step iterates are
\begin{align*}
    P(x,A)=\mu(Ax^{-1}),&\quad
P^n(x,A)=\mu^{*n}(Ax^{-1}),\\
A\in\mathcal B(G),&\quad n\ge1,
\end{align*}
where $\mathcal B(G)$ is the Borel $\sigma$-algebra of $G$.
The walk is \emph{open-set irreducible} if, for every $x\in G$
and every nonempty open set $U\subseteq G$,
\[
\sum_{n\ge1}P^n(x,U)
=
\sum_{n\ge1}\mu^{*n}(U x^{-1})
>0.
\]
Equivalently, for every such $x$ and $U$, there exists $n\ge1$
such that $\mu^{*n}(U x^{-1})>0$. 

Open-set irreducibility is also called topological irreducibility.
\end{definition}

\begin{definition}[Asymptotic uniformity]
\label{def:asymptotic_uniformity}
Let $G$ be a compact Hausdorff group with normalized Haar measure $\mu_G$, and let $\mu$ be a regular Borel probability measure on $G$. The random walk with increment law $\mu$ is \emph{asymptotically uniform in total variation} if
\[
\lim_{n\to\infty}
\|\mu^{*n}-\mu_G\|_{\mathrm{TV}}=0.
\]
\end{definition}

\subsection{Lie algebras and Lie groups}

\begin{definition}[Adjoint representation]
    $\Ad_X Y = X Y X^{-1}$ is the group adjoint representation. $\ad_X Y = [X, Y] = XY - YX$ is the Lie algebra adjoint representation, where the definition in terms of the commutator or Lie bracket holds in the case of a matrix Lie algebra.
\end{definition}

\begin{lemma}[Campbell, Ref.~\cite{campbell1897}]\label{lem:campbell}
 $\Ad_{e^{(X)}} = e^{\ad_X}$.
\end{lemma}

\begin{definition}[Generated group and closure]
    For a subset $S$ of a group $G$, $\langle S \rangle$ denotes the smallest subgroup containing $S$, and $\overline{\langle S \rangle}$ its closure in the topology of $G$.
\end{definition}

\begin{definition}[Lie algebra and Lie group{\rm\cite[Def.~1.1,2.0]{broker1985}}]\label{def:liealgebra}
A \emph{Lie algebra} is a real vector space $\mathfrak g$ with a bilinear map $[\cdot,\cdot] : \mathfrak g \times \mathfrak g \to \mathfrak g$ that is antisymmetric, $[X,Y] = -[Y,X]$, and satisfies the Jacobi identity $[X,[Y,Z]] + [Y,[Z,X]] + [Z,[X,Y]] = 0$. It is an algebra equipped with the Lie bracket. A \emph{Lie subalgebra} is a subspace closed under the bracket.
A \emph{Lie group} is a differentiable manifold $G$ also a group such that group multiplication and the inverse map is differentiable.
The \emph{Lie algebra of the Lie group $G$} is $\operatorname{Lie}(G) := T_e G$, equipped with the bracket inherited from $G$. 
\end{definition}

\begin{definition}[Generated Lie subalgebra {\rm \cite[p.14]{cap2026liegroups}, and \cite{androma_generated_lie}}]\label{def:gen_lie_sub}
Write  the set $\mathcal S \subseteq \mathfrak k$, where $\mathfrak k$ is a Lie algebra. The Lie subalgebra of $\mathfrak k$ generated by $\mathcal S$ is
\begin{multline}\label{eqn:generated_subalgebra}
     \Lie_\R\langle \mathcal S\rangle
    :=
    \bigcap
    \Bigl\{
        \mathfrak l\leq \mathfrak k: \\
        \mathfrak l\text{ is a real Lie subalgebra and }
       \mathcal  S\subseteq\mathfrak l
    \Bigr\}
\end{multline}
Note that a Lie subalgebra must be closed under the Lie brackets and be a vector space, both of which taking the intersection ensures.
\end{definition}

\begin{definition}[Lie subgroup-Lie subalgebra correspondence{\rm~\cite[Thm.~20.13]{lee2012}}]\label{def:subalgebra}
    Given a Lie subalgebra $\h$ of the Lie algebra $\g$ of a Lie group $G$, there exists a unique connected immersed Lie subgroup $H_{\h}\subseteq G$ whose tangent space at the identity is $\h$.
\end{definition}

\begin{theorem}[Closed subgroup theorem (Cartan)
  {\rm\cite[Thm.~3.11]{broker1985}}]\label{thm:cartan}
Every subgroup of a Lie group $G$ that is closed in $G$ is automatically an embedded Lie subgroup.
\end{theorem}

\subsection{Real Analysis}
\begin{theorem}[Real identity theorem]\label{thm:real_identity}
    A real-analytic function on $\R$ that vanishes on an interval of positive length is exactly zero everywhere~\cite[Thm.~8.5]{rudin1976}.
\end{theorem}

\begin{proposition}[Real analytic zero-set]{\rm~\cite[Prop.~0]{mityagin2020}}\label{prop:real_analytic_zero}
Let $A(x)$ be a real analytic function on (a connected open domain $U$ of) $\mathbb{R}^d$. If $A$ is not identically zero, then its zero set $\{ x \in U : A(x) = 0 \}$ has a zero measure.
\end{proposition}

\begin{definition}[$\mathrm{ess \, sup}${\rm~\cite[Def.~7.3, pg.80]{hunter2011}}]\label{def:ess_sup}
For a measurable function $f$ on a measure space $(X,\Sigma,\nu)$, its essential supremum is
\begin{gather*}
\operatorname*{ess\,sup}_{x\in X}f(x)
:=
\inf\left\{a\in\mathbb R:
\nu\bigl(\{x:f(x)>a\}\bigr)=0
\right\}.
\end{gather*}
It is the smallest upper bound that holds almost everywhere, allowing exceptions on a set of measure zero. If no finite such bound exists, it is \(+\infty\).
\end{definition}

\begin{definition}[Lipschitz map]
    A Lipschitz map $f$ has $ \lVert f(x) - f(y)\rVert \leq L \lVert x - y \rVert $ for some constant $L \geq 0$.
\end{definition}

\subsection{Riemannian geometry}
\begin{theorem}[Coarea formula (Thm.~7.3) in Ref.~\cite{simon2014}]\label{thm:coarea}
The idea is to express an integral over a region as the sum of integrals over the level sets of $f$.

Let \(M\) and \(N\) be smooth Riemannian manifolds with \(\dim M=m\ge n=\dim N\), and let \(f:M\to N\) be locally Lipschitz. For every nonnegative Borel measurable function \(g:M\to[0,\infty]\),
\begin{multline*}
\int_M g(x)\,J_n f(x)\,d\vol_M(x) \\
=
\int_N
\left(
\int_{f^{-1}(y)}g(x)\,d\Hd^{m-n}(x)
\right)d\vol_N(y).
\end{multline*}
Here
\begin{gather*}
J_n f(x)=\sqrt{\det\!\left(df_x\circ df_x^\dagger\right)}
\end{gather*}
is the \(n\)-dimensional Jacobian, defined almost everywhere, and \(\Hd^{m-n}\) is the normalized Hausdorff measure induced by the Riemannian distance on \(M\).

Also, writing \(\mathcal R=\{x\in M:df_x\text{ exists and }J_n f(x)>0\}\), we have:
\begin{multline*}
\int_{\mathcal{R}} g(x)\,d\vol_M(x) \\
=
\int_N
\left(
\int_{f^{-1}(y)\cap \mathcal R}
\frac{g(x)}{J_n f(x)}\,d\Hd^{m-n}(x)
\right) \\ d\vol_N(y).
\end{multline*}
If \(J_n f>0\) almost everywhere on \(M\), the integral on the left can be taken over all of \(M\).
\end{theorem}

\section{Necessary prior results}\label{apdx:nec_prior_results}
Here we provide for the convenience of the reader the two main results that we use, and did not ourselves prove, 
to complete the proof of \Cref{thm:exp_conv}.

\begin{theorem}[Reproduced from Bhattacharya~\cite{bhattacharya1972}, Theorem 3]\label{thm:bhattacharya}
    Given a regular Borel probability measure $\mu$ on a compact, Hausdorff, connected group $G$ that has a non-zero absolutely continuous component with respect to the Haar measure $\mu_G$ on $G$, then $d_{\mathrm{TV}}(\mu^{*n},  \mu_G)$ goes to zero exponentially fast (in total variation distance) as $n \to \infty$.
\end{theorem}

\begin{corollary}[Reproduced from Anoussis--Gatzouras~\cite{anoussis2004}, Corollary~4.2.]\label{cor:anoussis}
Suppose $G$ is a compact group, and $\mu$ a regular Borel probability measure on $G$.
Suppose further that $\mu$ is absolutely continuous with respect to Haar measure
$\mu_G$ on $G$. Then
\[
  \lim_{n \to \infty} \|\mu^{*n} - \mu_G\|_\mathrm{TV}^{1/n}
  = \max_{R \in \Grep \setminus \{\mathbf{1}_G\}} \operatorname{spr}\bigl(\hat{\mu}(R)\bigr),
\]
where $\Grep$ is the unitary dual of $G$, i.e. the collection of irreducible representations of $G$ equivalent up to conjugation by a unitary, and $\mathrm{spr}(M)$ is the spectral radius of a complex square matrix $M$ (\Cref{def:spr})
If in particular $(G, \mu)$ is adapted, i.e. is not supported by a proper closed subgroup of $G$, and $\mu$ strictly aperiodic, i.e. not concentrated on a coset of a proper, closed, normal subgroup of $G$, then
$\|\mu^{*n}-\mu_G\|_{\mathrm{TV}} \to 0$ exponentially fast.
\end{corollary}

For a discussion on the meaning of the representation theoretic ideas involved in \Cref{cor:anoussis}, see Appendix \Cref{apdx:rep_theory}.

\section{Representation theory tutorial}\label{apdx:rep_theory}
Here we give a brief discussion of the ideas from representation theory needed to understand \Cref{cor:anoussis}.
Let $d_R \geq 1$. Define the continuous map
\begin{gather*}
    R \, : \, G \to U(d_R), \; \; R(hg) = R(h)R(g) \; \forall \, h, g \in G
\end{gather*}
be a continuous unitary representation of the compact finite dimensional group $G$.
$R$ is called irreducible if there does not exist a complex linear subspace $X$ such that $\{ 0 \} \subsetneq X \subsetneq \C^{d_R}$ such that $R(g) \chi \in X$ for every $\chi \in X$, $g \in G$.
Let $\Grep$ be a set of all the irreps of $G$ up to unitary equivalence: $R_2(g) = \Ad_U R_1(g)$ for all $g \in G$ \cite[Thm.~5.2]{folland2016} and a fixed $U \in U(d_R)$. $\Grep$ is called the unitary dual of $G$. 

For each $R \in \Grep$, the Fourier coefficient, which is a matrix averaged entry by entry of unitaries, is:
\begin{align}
    \widehat{\mu}(R)
    &:=\int_G R(g^{-1})\,d\mu(g)
    \label{eqn:muhat}\\
    &=\frac1{\vol(\D)}
    \int_{\D_A}\int_{\D_B}
    R(F(\phi,t)^{-1})\,dt\,d\phi.
    \nonumber
\end{align}
By the representation property $R(gh) = R(g)R(h)$ and the independence of each convolution factor,
\begin{gather*}
    \widehat{\mu^{*L}}(R) = \widehat{\mu}(R)^{L}.
\end{gather*}
The trivial irrep $\mathbf{1}_G$ has $\widehat{\mu}(\mathbf{1}_G)=1$, and any non-trivial irrep has
\begin{gather*}
  \int_G R(g^{-1}) \, d\mu_G(g) = 0.
\end{gather*}
We can see this via a proof by contradiction. 
Let 
\begin{gather*}
    P:= \int_G R(g^{-1}) \, d\mu_G(g)
\end{gather*} 

For every fixed $h\in G$, take $u=gh^{-1}$, so that $u^{-1}=hg^{-1}$. By right invariance of Haar measure,
\begin{align*}
R(h)P
&=\int_G R(hg^{-1})\,d\mu_G(g)\\
&=\int_G R(u^{-1})\,d\mu_G(u)
=P.
\end{align*}

Then, $R(h)Pv = Pv := w$ for any vector $v$. If $w \neq 0$, then any $w' \in X = \Span_\C \{ w \}$ has $R(h) w' \in X$. Hence $X$ is a non-zero invariant subspace.
Since $R$ is irreducible, 
\[
X=\C^{d_R}.
\]
But $\dim_\C X = 1$, so $d_R = 1$. Since $R(h)$ fixes every $v \in X$, for all $h \in G$, $R(h)=1$.
That makes $R$ the trivial representation, contradicting our assumption.

By the linearity of integration, for a non-trivial irrep $R$,
\begin{gather*}
    \widehat{\mu-\mu_G}(R) = \widehat{\mu}(R) - \widehat{\mu_G}(R) = \widehat\mu(R).
\end{gather*}
For the trivial irrep, $\widehat{\mu}(R) = 1 = \widehat{\mu_G}(R)$.

\begin{definition}\label{def:spr}
    The spectral radius for a complex matrix $M$ is defined by:
\begin{gather*}
    \operatorname{spr}(M) := \max \{ |E| \, : \, E \in \mathrm{spec}(M)\}.
\end{gather*}
\end{definition}

\section{Supplemental proofs}\label{apdx:extra_proofs}

\begin{proposition}[$\overline{\Ucal}$ is a group]\label{prop:G_group}
As defined in \Cref{eqn:G}, $G$ is a compact Lie subgroup of $\mathrm{U}(d)$.
\end{proposition}
\begin{proof}
Let $x \in \D^{L_1}$, giving $X = U_{L_1}(x)$. Similarly, $y \in \D^{L_2}$, with $Y = U_{L_2}(y)$.
Then,
\begin{multline*}\label{eqn:concatenate}
    U_{L_1}(x)U_{L_2}(y) \\ = U_{L_1 + L_2}((x,y)) \in \Ucal_{L_1 + L_2},
\end{multline*}
where $(x,y)$ denotes the concatenation preserving the order of $x$ and $y$.
Since this holds for all $L_1, L_2 \geq 1$ and all $X \in \Ucal_{L_1} , Y \in \Ucal_{L_2}$, $\Ucal$ is closed under matrix multiplication. Matrix multiplication is associative, so $\Ucal$ is a semigroup under matrix multiplication. 

Choose some $Z, W\in G=\overline{\Ucal}$. Consider the sequences 
\begin{gather*}
(Z_m)_{m \geq 0},\; (W_m)_{m \geq 0} \subseteq \Ucal
\end{gather*}
such that $Z_m\to Z$ and $W_m\to W$.
Since $\Ucal$ is a semigroup, for every $m$,
\begin{gather*}
Z_mW_m\in\Ucal 
\end{gather*}
The continuity of matrix multiplication gives
\begin{equation*}
     Z_m W_m\longrightarrow Z W \Longrightarrow ZW \in \overline{\Ucal} = G,
\end{equation*}
establishing that $G$ is a semigroup.

It remains to show that $G$ contains all its inverses. 
We know that $G$ is a compact set, because it is closed in the compact topological group $\mathrm{U}(d)$. $\mathrm{U}(d)$ is also a metric space, so compactness here means that every sequence in $G$ has a convergent subsequence~\cite[Prop.A.18, p.553]{lee2012}. 

That is, consider the sequence of powers of some $Z \in G$, $(Z^m)_{m \geq 1}$. 
Compactness of $G$ implies there exists a strictly increasing sequence $m_0 <  m_1 < \cdots $ for $m_\ell \in \mathbb{N}$ such that 
\begin{equation*}
    Z^{m_\ell} \xrightarrow{\ell \to \infty} W, \quad W \in G.
\end{equation*}
Since $m_{\ell+1} - m_\ell \geq 1$ and $W^{-1}, Z^{-1} \in \mathrm{U}(d)$ exist, 
\begin{multline*}
    \forall \ell\geq 0, \; G \ni Z^{m_{\ell+1} - m_\ell-1}  = \left ( Z ^{m_{\ell}}  \right )^{-1} Z^{m_{\ell+1}} Z^{-1} \\
    \xrightarrow{\ell \to \infty} W^{-1} W Z^{-1} = \I Z^{-1} = Z^{-1}.
\end{multline*}

Since $G$ is closed, it contains the limits of all convergent sequences whose terms lie in $G$, so $Z^{-1} \in G$. Since $Z \in G$ was arbitrary, $G$ contains all its inverses. Since $0 \in \D_B$, $\I \in \Ucal \subseteq G$. $G$ is therefore a closed subgroup of $\mathrm{U}(d)$. 
By Cartan's closed subgroup theorem~\cite[Thm.20.10, p.526]{lee2012} (\Cref{thm:cartan} in Appendix \Cref{apdx:def}), $G$ is an embedded Lie group in $\mathrm{U}(d)$.
\end{proof}

The next three results concern properties of the Jacobian map, which are essential for the proof of the absolute continuity of the pushforward measure.

\begin{lemma}[Helper lemma, real analyticity of the squared normal Jacobian]\label{lem:jacobian_real_analytic}
The squared normal Jacobian $\Lambda_L(x)$ is real analytic.
\end{lemma}
\begin{proof}
$U_L$ is a real analytic function because the matrix exponential and finite matrix products are real analytic. Hence the entries of $DU_L(x)$ are real analytic functions of $x$. Since matrix inversion is real analytic on $\mathrm{U}(d)$, the entries of the right-trivialized differential $J_{L,x}(v)$ are real analytic. Taking the transpose, multiplying matrices, and taking the determinant preserve real analyticity. 
\end{proof}

\begin{proposition}[Submersion somewhere means submersion almost everywhere]\label{prop:submersion_almost_everywhere}
    Let $x \in \D^L$ be a submersion point of $U_L$. Then the set of points in $\D^L$ that are not submersion points has measure zero with respect to the normalized Lebesgue measure $\lambda^{\otimes L}$.
\end{proposition}
\begin{proof}
$\Lambda_L(x)$ is real analytic by \Cref{lem:jacobian_real_analytic}, but $\D^L$ need not be connected by \Cref{def:main}. However, the same expressions for $U_L$ and the squared normal Jacobian on $\D^L$ also define real-analytic functions on the connected open set $\R^{2L}$.
We denote the resulting extension by 
\begin{equation*}
    \widetilde{\Lambda}_L\, :\, \mathbb R^{2L}\to\R.
\end{equation*}

If there exists one $x\in \D^L$ such that  $\widetilde{\Lambda}_L(x) \neq 0$, then $\widetilde{\Lambda}_L$  is not identically zero. Hence by \Cref{prop:real_analytic_zero} (\cite[Prop.~0]{mityagin2020}), the zero set of $\widetilde{\Lambda}_L$ has Lebesgue measure zero in $\R^{2L}$.

Consequently,
\[
    \mathcal E_L:=
    \{x\in(\D^\circ)^L:\Lambda_L(x)=0\}
\]
has Lebesgue measure zero. Moreover,
\[
    \D^L\setminus\mathcal N_L
    =
    \bigl(\D^L\setminus(\D^\circ)^L\bigr)\cup \mathcal E_L.
\]

Every point of $\D^L\setminus(\D^\circ)^L$ has at least one time coordinate equal to zero. Fixing a coordinate to a single value gives a set of zero Lebesgue measure, and the finite union over the $L$ time coordinates still has measure zero. Since $\mathcal E_L$ also has measure zero, their union $\D^L\setminus\mathcal N_L$ has measure zero.
Normalizing Lebesgue measure therefore gives
\[
    \lambda^{\otimes L}
    \bigl(\D^L\setminus\mathcal N_L\bigr)=0.
\]
\end{proof}

\begin{lemma}[Helper lemma, range of Jacobian]\label{lem:im_J_subset_h}  
Define
\begin{gather*}
\operatorname{im}J_{L,x} =  \left\{ DU_L(x)[v]U_L(x)^{-1} \, : \, \ v\in T_x\mathcal D^L \right\} \\
= DU_L(x)\bigl[T_x\mathcal D^L\bigr]U_L(x)^{-1}.
\end{gather*}
For every $L\geq1$ and $x\in\mathcal D^L$,
\begin{equation*}
    \operatorname{im}J_{L,x} \subseteq \h.
\end{equation*}
\end{lemma}
\begin{proof}
By definition, $C(\phi)\in\mathfrak h$ for every
$\phi\in\D_A$. The map
\begin{gather*}
    \D_A\longrightarrow\mathfrak h,
    \quad
    \phi\longmapsto C(\phi),
\end{gather*}
is smooth. Therefore $ F(\phi,\timeB) =e^{\timeB C(\phi)}$ defines a smooth map from $\mathcal D$ into $H_{\mathfrak h}$.
Since multiplication in a Lie group is smooth, the product map
\begin{gather*}
    U_L:\mathcal D^L\longrightarrow H_{\mathfrak h},
\end{gather*}
is smooth. Then we can apply the definition of the differential of a smooth function:
\begin{equation*}
     DU_L(x) \, : \,T_x\D^L \longrightarrow T_{U_L(x)} H_\h.
\end{equation*}
That is,
\begin{equation}
\begin{gathered}
   \forall \, v \in T_x \D^L, \; DU_L(x)[v] \in T_{U_L(x)} H_\h \\
   \Longrightarrow \quad  DU_L(x)\bigl[T_x\D^L\bigr] \subseteq T_{U_L(x)} H_\h.
\end{gathered}
\end{equation}

Right translation by $U_L(x)^{-1}$  maps $T_{U_L(x)}H_{\mathfrak h}$ onto
$T_\I H_\h =\h,$ hence by the definition of the right-trivialized Jacobian in Def.~\eqref{def:right_trivialized_differential},
\begin{equation}
\begin{split}
     DU_L(x)\bigl[T_x\D^L\bigr]\, U_L(x)^{-1}  &= J_{L,x}\bigl(T_x\D^L\bigr) \\ &= \operatorname{im}J_{L,x} \subseteq  \h.
\end{split}
\end{equation}
\end{proof}

In \Cref{lem:comm_comm_sufficient} we give a sufficient condition that ensures $\h = \g$, as required by the main results. It says that any matrix that commutes with both $A$ and $B$ must be proportional to the identity.

\begin{definition}[Common centralizer]
The common centralizer of $A$ and $B$ as
\begin{align*}
\Comm &:=\! \{Q\in M_d(\mathbb C):[Q,A]=[Q,B]=0\} \\
&= \mathcal Z(A) \cap \mathcal Z(B),
\end{align*}
where $\mathcal{Z}(A)$ is the centralizer of the one parameter subgroup $\{ e^{i \phi A}\} $, and analogously for $B$.
\end{definition}

\begin{proposition}[Sufficient condition for $\h=\g$]\label{lem:comm_comm_sufficient}
\begin{gather*}
\Comm=\C \I \quad \Longrightarrow\quad \h=\g.
\end{gather*}
\end{proposition}

\begin{proof}
By \Cref{prop:subalgebra}, $\overline{H_\h}=G$. 
Let $a \in G$ be arbitrary. Since $\overline{H_\mathfrak h}=G$, there exists a sequence $a_n\in H_\h$ such that $a_n \xrightarrow{n \to \infty} a$. Then for every $Y \in \h$
\begin{gather*}
 \Ad_{a_n}  Y  \xrightarrow{n \to \infty} \Ad_{a}  Y \in \h,
\end{gather*}
since conjugation by $H_\h$ preserves its Lie algebra $\h$, and $\h$ is a closed vector subspace, and $\Ad$ is continuous. 

For arbitrary $X\in\g$, apply this to $a=e^{\timeB X}$ for $t \in \R$ and differentiate at $\timeB=0$:
\begin{equation}
    \frac{d}{d\timeB}\bigg |_{\timeB=0} \left ( \Ad_{e^{\timeB X}} Y \right ) = [X, Y] \in \h \quad \Longrightarrow \quad [\g, \h] \subseteq \h,
\end{equation}
i.e. $\h$ is an ideal of $\g$ because $X, Y$ were arbitrary. 

Since $\h \subseteq \g$, we can define its orthogonal complement $\h^{\perp_\g}$ in $\g$
\begin{equation}
    \h^{\perp_\g} := \{X \in \g \, : \, \langle X, Z \rangle = 0 \, \text{ for every } Z \in \h\}.
\end{equation}
As a vector space,
$
\g = \h \oplus \h^{\perp_\g}.
$

Recall $\Ad_{e^{tX}}$ preserves the Hilbert-Schmidt inner product. Let $X, U, V \in \g$:
\begin{align*}
\langle \Ad_{e^{\timeB X}}U, \Ad_{e^{\timeB X}}V  \rangle &= \operatorname{Re}\tr \left [ e^{\timeB X}U V^\dagger e^{-\timeB X}  \right ] \\ &= \langle U, V \rangle.
\end{align*}
Take the derivative at $\timeB=0$:
\begin{gather*}
\langle [X, U], V  \rangle + \langle U, [X, V]  \rangle= 0.
\end{gather*}
Specifically, choose $X \in \g$, $Y \in \h$, and $Z \in \h^{\perp_\g}$. We have $\langle Z, [X, Y] \rangle = 0$ because $[X,Y] \in \h$. Then,
\begin{gather*}
\langle [X, Z], Y  \rangle = - \langle Z, [X, Y]  \rangle = 0 \\
\Longrightarrow  \quad [X, Z] \in \h^{\perp_\g} \quad \Longrightarrow \quad [\h^{\perp_\g}, \g] \subseteq \h^{\perp_\g}.
\end{gather*}
Thus $\h^{\perp_\g}$ is an ideal of $\g$, because $X, Y, Z$ were arbitrary. Hence $[\h, \h^{\perp_\g}] \subseteq \h \cap \h^{\perp_\g} = \{ 0\}$. Applying this property to $Y \in \h, \, Z \in \h^{\perp_\g}$
\begin{gather*}
[Y,Z]=0 \quad\Longrightarrow\quad \Ad_{e^{\timeB Y}}Z=e^{\timeB\ad_Y}Z=Z, \\ \timeB\in\R.
\end{gather*}
Since the connected group $H_\h$ is generated by exponentials of elements of $\h$, $\Ad_{a_n}Z=Z$ for every $a_n\in H_\h$. 
For every $b\in\g$ and $\timeB\in\R$, choose $a_n\in H_\h$ with $a_n\to e^{\timeB b}$. By continuity,
\begin{gather*}
\operatorname{Ad}_{e^{\timeB b}}Z
=\lim_{n\to\infty}\operatorname{Ad}_{a_n}Z=Z.
\end{gather*}
Differentiating at $\timeB=0$:
\begin{gather*}
    \frac{d}{dt}\bigg|_{t=0} \left(\Ad_{e^{tb}}Z\right)=[b,Z]=0.
\end{gather*}
Since $Z$ was any element in $\h^{\perp_\g}$, 
\begin{equation}\label{eqn:hperp_g_subset}
    \h^{\perp_\g} \subseteq \mathfrak z(\g) := \{ U \in \g \, : \, [U, W] = 0 \, \text{ for any } W \in \g\},
\end{equation} 
where $\mathfrak z(\g)$ is the center of $\g$.

Define
\[
\mathcal M
=
\{e^{i\phi A} : \phi\in\mathbb R\}
\cup
\{e^{itB} : t\in\mathbb R\},
\]
and let $T$ be the closed subgroup generated by $\mathcal M$:
\begin{gather*}
T
=
\overline{
  \{U_0\cdots U_{n-1} :
    n\geq 1,\;
    U_0,\ldots,U_{n-1}\in\mathcal M\}
},\\
\tf=\operatorname{Lie}(T),
\end{gather*}
where the empty product is the identity and the closure is taken in $U(d)$.
Since $F(\phi,\timeB)=e^{i\phi A}e^{i\timeB B}e^{-i\phi A} \in T$ for every $\phi,\timeB\in\R$,
we have $G\subseteq T$, and hence $\g\subseteq\tf$.

We next show that restricting $\phi$ to $\D_A$ does not change the real span of the conjugates
$C(\phi)=ie^{i\phi A}Be^{-i\phi A}$.
Set
\[
W=\operatorname{span}_{\mathbb R}
  \{C(\phi):\phi\in\D_A\}.
\]
Work in the real vector space of skew-Hermitian matrices, with the real Hilbert--Schmidt inner product.
Let \(E_1,\ldots,E_m\) be an orthonormal basis of \(W^\perp\).
The orthogonal projection of \(C(\phi)\) onto \(W^\perp\) is
\[
P_{W^\perp}C(\phi)
=
\sum_{j=1}^{m} c_j(\phi)E_j,
\quad
c_j(\phi)=\langle E_j,C(\phi)\rangle.
\]
The functions $c_j$ are the coefficients of this projection in the chosen basis.

The matrix-valued map $\phi\mapsto C(\phi)=ie^{i\phi A}Be^{-i\phi A}$ is real analytic. Since taking the inner product with a fixed matrix $E_j$ is a real-linear functional, each $c_j$ is a real analytic function on $\R$.
For every $\phi\in\D_A$, the definition of $W$ gives $C(\phi)\in W$, so $c_j(\phi)=0$ for every $j$. Because $\D_A$ contains an interval of positive length, the identity theorem for real analytic functions (\Cref{thm:real_identity}) implies that each $c_j$ vanishes on $\R$. Consequently, $P_{W^\perp}C(\phi)=0$, and hence $C(\phi)\in W$, for every $\phi\in\mathbb R$.
Together with the inclusion $\D_A\subseteq\mathbb R$,
\[
\operatorname{span}_{\mathbb R}
  \{C(\phi):\phi\in\D_A\}
=
\operatorname{span}_{\mathbb R}
  \{C(\phi):\phi\in\mathbb R\}.
\]
Since $W\subseteq\h\subseteq\g$, it follows that $C(\phi)\in\h\subseteq\g$ for every $\phi\in\R$. 
The reasoning behind \Cref{thm:real_identity} is that an analytic function equals its Taylor series near every point. If it has been zero on an interval, all its derivatives are zero at the interval's edge, so its Taylor series forces it to remain zero just beyond that edge. It therefore cannot start becoming nonzero outside this interval.

Thus, restricting the parameter set to a nonempty finite union of nondegenerate intervals preserves the real span of these conjugates and, consequently, the Lie algebra they generate. The size of the interval may nevertheless affect the rate of convergence to Haar measure, but we leave that to future work.

For any $a\in G=\overline{\Ucal}$, choose $a_n\in\Ucal$ with $a_n \to a$. By continuity of conjugation and closedness of $G$,
\begin{gather*}
e^{i\phi' A}ae^{-i \phi' A}
=\lim_{n\to\infty}e^{i\phi' A}a_ne^{-i\phi' A}\in G.
\end{gather*}
Therefore $\Ad_{e^{i\phi' A}}G\subseteq G$.

Moreover, $\Ad_{e^{i \timeB' B}} a \in G$ for $a \in G$ and $\timeB\in \R$ because $e^{i \timeB' B} \in G$. For any $\ell \in T$, define the sequence  for $(\ell_n)_{n \geq 0}$ where each $\ell_n\in \langle e^{i \phi A}, e^{i \timeB B} \, : \, \phi, \timeB\in \R\rangle$ such that
\begin{gather*}
\ell_n \xrightarrow{n \to \infty} \ell.
\end{gather*}
For each $a \in G$, $\Ad_{\ell_n} a \in G \to \Ad_\ell a \in G$. That is,
\begin{gather*}
\ell G \ell^{-1} \subseteq G.
\end{gather*}
We make the same argument replacing $\ell \to \ell^{-1}$ and obtain the reversed inclusion, hence $\Ad_\ell G = G$.

Taking the tangent space at the identity, $\Ad_\ell \g = \g$. With $X \in \g$ and $K \in \tf$, 
\begin{gather*}
 \frac{d}{d\timeB}\bigg |_{\timeB=0} \left ( \Ad_{e^{\timeB K}} X  \right ) = [K, X]  \in \g \quad \Longrightarrow \quad [\tf, \g] \subseteq \g.
\end{gather*}
Define  
\begin{equation}
    \g^{\perp_\tf} := \{X \in \tf \, : \, \langle X, Z \rangle = 0 \, \text{ for every } Z \in \g\}.
\end{equation}
Hence we have the orthogonal decomposition $\tf = \g \oplus \g^{\perp_\tf}$.
As before, we can show $\g^{\perp_\tf}$ is also an ideal of $\tf$.
Therefore, $[\g, \g^{\perp_\tf}] = \{ 0 \}$.

Take $Z \in \mathfrak z(\g)$. Let  
\begin{gather*}
\tf \ni V = \underbrace{X}_{\in \g} +\underbrace{U}_{\in \g^{\perp_\tf}}.
\end{gather*}
$[Z, X] = 0$ because $Z \in \mathfrak z(\g)$. $[Z, U] = 0$ because $Z \in \g$, $U\in\g^{\perp_\tf}$, and these two ideals commute.
Then,
\begin{gather*}
[Z, V] = [Z, X] + [Z, U] = 0 \quad \Longrightarrow \quad \mathfrak z (\g) \subseteq \mathfrak z(\tf),
\end{gather*}
where
\begin{gather*}
\mathfrak z(\tf) := \{P \in \tf \, : \,   [P, Q ] = 0 \, \text{ for any } Q \in \tf \}.
\end{gather*}
Combining with \Cref{eqn:hperp_g_subset}, we have that
\begin{equation}\label{eqn:comm_comm_subset}
    \h^{\perp_\g} \subseteq \mathfrak z(\g) \subseteq \mathfrak z(\tf) \subseteq \Comm \cap \tf,
\end{equation}
where $\Comm$ is the common centralizer. The final inclusion holds because $iA,iB \in \tf$, so every element of $\mathfrak z(\tf)$ commutes with $A$ and $B$. 

By hypothesis, $\Comm = \C \I$. Then by \Cref{eqn:comm_comm_subset}, 
\begin{gather*}
\h^{\perp_\g} \subseteq i \R \I,
\end{gather*}
since $\tf\subseteq\mathfrak \mathrm{U}(d)$ consists of skew-Hermitian matrices. Because $i \R \I$ is a one dimensional real linear vector space, it has only two linear subspaces: $\{ 0\}$ and itself. Suppose $\h^{\perp_\g} \neq \{ 0 \}$. 
Then $\h^{\perp_\g} = i \R \I$. Since $i B \in \h$,
\begin{gather*}
0 = \langle i  \I, i B \rangle = \tr(B) \Longrightarrow \g \subseteq \mathfrak{su}(d) \Longrightarrow i I \notin \g.
\end{gather*}
Then, due to the fact that for any matrix $A$, $\det (e^A) = e^{\tr(A)}$:
\begin{gather*}
\det F(\phi,\timeB)
=e^{\timeB\tr(C(\phi))}
=e^{i\timeB\tr(B)}=1.
\end{gather*}
Hence $\Ucal\subseteq \mathrm{SU}(d)$, and since $\mathrm{SU}(d)$ is closed, $G\subseteq \mathrm{SU}(d)$. Hence $\g\subseteq\mathfrak{su}(d)$.

However, this contradicts $i \I \in \h^{\perp_\g} \subseteq \g$. Therefore 
\begin{gather*}
\h^{\perp_\g} = \{ 0 \} \quad \Longrightarrow \quad \g = \h.
\end{gather*}
\end{proof}

\section{Examples of our mechanism}

This appendix illustrates how varying a single control parameter produces a family of Hamiltonians related by conjugation,
\begin{equation}
    H(\phi)=e^{i\phi A}Be^{-i\phi A},
\end{equation}
possibly with an additional scalar rescaling or a change of frame.
We begin with one- and two-qubit examples, then consider Rydberg arrays and Hubbard models.
Throughout, we set $\hbar=1$.
For physical evolution $e^{-itH(\phi)}$, the main results apply but we use the convention of a negative sign in the exponential, i.e. $e^{- i t H}$.

\subsection{One- and two-qubit examples}\label{apdx:lin_ind}

\paragraph*{Single-qubit phase control.}
Consider
\begin{equation}
    H_0(\theta)
    =\Omega(\cos\theta\,\sigma^x+\sin\theta\,\sigma^y)
    +\delta\sigma^z,
\end{equation}
where $\theta$ is adjustable and $\Omega,\delta\in\R$ are fixed.
For two settings $\theta_0$ and $\theta_1$,
\begin{align*}
    [H_0(\theta_0),H_0(\theta_1)]
    =2i\bigl[&\Omega\delta(\sin\theta_0-\sin\theta_1)\sigma^x
    \nonumber\\
    &+\Omega\delta(\cos\theta_1-\cos\theta_0)\sigma^y
    \nonumber\\
    &+\Omega^2\sin(\theta_1-\theta_0)\sigma^z\bigr].
\end{align*}
If $\Omega\delta\ne0$, this commutator is nonzero precisely when $\theta_1-\theta_0\notin2\pi\mathbb Z$.
If $\delta=0$ and $\Omega\ne0$, the condition is instead $\theta_1-\theta_0\notin\pi\mathbb Z$; if $\Omega=0$, all settings commute.
Whenever the commutator is nonzero, the three skew-Hermitian operators
\begin{equation*}
    iH_0(\theta_0),\quad iH_0(\theta_1),\quad
    [H_0(\theta_0),H_0(\theta_1)]
\end{equation*}
are linearly independent over $\mathbb R$ and span $\mathfrak{su}(2)$.
Indeed, the Pauli-vector representation identifies the commutator with the cross product of the two Hamiltonian vectors.

The conjugation structure follows by defining
$R_z(\theta)=e^{-i\theta\sigma^z/2}$:
\begin{equation}\label{eqn:change_phi}
    H_0(\theta)
    =R_z(\theta)(\Omega\sigma^x+\delta\sigma^z)R_z(\theta)^\dagger.
\end{equation}
Thus $A=-\sigma^z/2$ and $B=\Omega\sigma^x+\delta\sigma^z$.\\

\paragraph*{Single-qubit detuning control.}
Now let the detuning vary while $\Omega\ne0$ remains fixed:
\begin{equation}
    H_1(\delta)=\delta\sigma^z+\Omega\sigma^x.
\end{equation}
Two distinct settings satisfy
\begin{equation*}
    [H_1(\delta_0),H_1(\delta_1)]
    =2i\Omega(\delta_0-\delta_1)\sigma^y\ne0,
\end{equation*}
and therefore generate $\mathfrak{su}(2)$ as above.
Define $r(\delta)=\sqrt{\delta^2+\Omega^2}$ and choose $\theta(\delta)$ so that
\begin{equation*}
    \cos\theta(\delta)=\frac{\delta}{r(\delta)},\qquad
    \sin\theta(\delta)=\frac{\Omega}{r(\delta)}.
\end{equation*}
With $R_y(\theta)=e^{-i\theta\sigma^y/2}$,
\begin{equation}\label{eqn:change_z}
    H_1(\delta)
    =r(\delta)R_y(\theta(\delta))\sigma^zR_y(\theta(\delta))^\dagger.
\end{equation}
Here $A=-\sigma^y/2$ and $B=\sigma^z$.
Unlike phase control, detuning control changes both the orientation and the norm of the Hamiltonian. For a fixed detuning during a segment of duration $t$, the norm change can be absorbed into the evolution time:
\begin{gather*}
    e^{-itH_1(\delta)}
    =R_y(\theta(\delta))e^{-i\tau\sigma^z}R_y(\theta(\delta))^\dagger,
    \; \tau=r(\delta)t.
\end{gather*}\\

\paragraph*{Two-qubit interaction control.}
Consider a tunable Ising-type interaction,
\begin{align}
    H(t)&=B+J(t)A,\\
    A&=\sigma_0^z\sigma_1^z,\qquad
    B=\sum_{i=0}^{1}(\delta_i\sigma_i^z+\Omega\sigma_i^x),
\end{align}
with fixed $\delta_i$ and $\Omega$.
Define
\begin{equation}
    \chi(t)=\int_0^t J(s)\,ds,\qquad V(t)=e^{i\chi(t)A}.
\end{equation}
For the rotating-frame state $\ket{\widetilde\psi(t)}=V(t)\ket{\psi(t)}$, the Hamiltonian is
\begin{align}
    \widetilde H(t)
    &=V(t)H(t)V(t)^\dagger+i\dot V(t)V(t)^\dagger
    \nonumber\\
    &=V(t)BV(t)^\dagger,
\end{align}
since $i\dot V(t)V(t)^\dagger=-J(t)A$.
Explicitly,
\begin{align}
    \widetilde H(t)
    ={}&\delta_0\sigma_0^z+\delta_1\sigma_1^z
    \nonumber\\
    &+\Omega\cos(2\chi(t))(\sigma_0^x+\sigma_1^x)
    \nonumber\\
    &-\Omega\sin(2\chi(t))
    (\sigma_0^y\sigma_1^z+\sigma_0^z\sigma_1^y).
\end{align}
Thus the interaction strength controls the rate $\dot\chi(t)=J(t)$ at which the conjugation parameter changes.
The laboratory-frame propagator is
\begin{equation}\label{eqn:squash_frame_endpoint}
    U(t)=e^{-i\chi(t)A}\widetilde U(t).
\end{equation}

\subsection{Many-body Hamiltonian examples}\label{apdx:examples}

We show how Rydberg arrays, Bose and Fermi Hubbard models, superconducting circuits, and trapped ions realize the conjugation family $H(\phi)=e^{i\phi A}Be^{-i\phi A}$, either directly or in an interaction picture.
Throughout, we take $\hbar=1$.

Applying our convergence results additionally requires the connected Lie subgroup generated by the conjugated Hamiltonians to be closed; a sufficient condition is that the common centralizer $\Comm=\{Q:[Q,A]=[Q,B]=0\}$ consists only of scalar multiples of the identity on the Hilbert space under consideration (\Cref{lem:comm_comm_sufficient}).
We defer a detailed analysis of this condition for the individual models.

\paragraph*{Rydberg atoms: drive phase control.}
An array of Rydberg atoms driven with a global Rabi frequency and phase has the Hamiltonian~\cite{matsoukasroubeas2026,wurtz2023}
\begin{align}
    H(\varphi)
    ={}&\frac{\Omega}{2}\sum_{i=0}^{N-1}
    \left(e^{i\varphi}\ket{r}_i\bra{g}_i
    +e^{-i\varphi}\ket{g}_i\bra{r}_i\right)
    \nonumber\\
    &-\sum_i\Delta_i\hat n_i+\sum_{i<j}J_{ij}\hat n_i\hat n_j,
\end{align}
where $\hat n_i=\ket{r}_i\bra{r}_i$, $\Delta_i$ are fixed detunings, and $J_{ij}$ are fixed van der Waals interaction strengths.
Let $\hat N=\sum_i\hat n_i$ be the total Rydberg excitation number.
Since $\hat n_i=\ket{r}_i\bra{r}_i$ is a projector,
\[
    e^{\pm i\phi\hat n_i}
    =\ket{g}_i\bra{g}_i
    +e^{\pm i\phi}\ket{r}_i\bra{r}_i.
\]
The number operators on distinct sites commute, so
\[
    e^{i\phi\hat N}=\prod_j e^{i\phi\hat n_j}.
\]
All factors with $j\ne i$ commute with
$\ket{r}_i\bra{g}_i$ and cancel in the conjugation.
Using $\langle g|r\rangle=0$, we therefore obtain
\begin{align*}
    &e^{i\phi\hat N}\ket{r}_i\bra{g}_i e^{-i\phi\hat N}\\
    &\quad=e^{i\phi\hat n_i}
    \ket{r}_i\bra{g}_i e^{-i\phi\hat n_i}\\
    &\quad=
    \left(\ket{g}_i\bra{g}_i
    +e^{i\phi}\ket{r}_i\bra{r}_i\right)\ket{r}_i\bra{g}_i\\
    &\qquad \qquad \left(\ket{g}_i\bra{g}_i
    +e^{-i\phi}\ket{r}_i\bra{r}_i\right)\\
    &\quad=e^{i\phi}\ket{r}_i\bra{g}_i.
\end{align*}
Taking the adjoint gives
\[
    e^{i\phi\hat N}\ket{g}_i\bra{r}_i e^{-i\phi\hat N}
    =e^{-i\phi}\ket{g}_i\bra{r}_i.
\]

The diagonal part
\begin{equation}
    D=-\sum_i\Delta_i\hat n_i
    +\sum_{i<j}J_{ij}\hat n_i\hat n_j
\end{equation}
commutes with $\hat N$, since all $\hat n_i$ commute.
Consequently,
\begin{align}
    e^{i\phi\hat N}H(\varphi)e^{-i\phi\hat N}
    &=
    \frac{\Omega}{2}\sum_i
    \bigg(
    e^{i\varphi}e^{i\phi}\ket{r}_i\bra{g}_i \nonumber\\
    & \quad +e^{-i\varphi}e^{-i\phi}\ket{g}_i\bra{r}_i
    \bigg)+D \nonumber\\
    &=H(\varphi+\phi).
    \label{eqn:aquila_u1covariance}
\end{align}

This realizes the $H(\phi) = e^{i \phi A} B e^{- i \phi A}$ family with $A=\hat N$ and $B=H(0)$, extending \Cref{eqn:change_phi} to $N$ qubits.

\paragraph*{Bose--Hubbard: interaction control.}
The same rotating-frame construction applies to a Bose-Hubbard model with a tunable on-site interaction \cite{dutta2015, pasienski2010}.
Suppose the spatially uniform on-site interaction strength $u(t)$ can be varied while the hopping amplitudes remain fixed:
\begin{align}
    H_{\mathrm B}(t)&=B_{\mathrm B}+u(t)A_{\mathrm B},\\
    B_{\mathrm B}&=-\sum_{\langle i,j\rangle}
    \left(J_{ij}\hat a_i^\dagger\hat a_j
    +J_{ij}^*\hat a_j^\dagger\hat a_i\right),\\
    A_{\mathrm B}&=\frac12\sum_i\hat n_i(\hat n_i-1),
    \qquad \hat n_i=\hat a_i^\dagger\hat a_i.
\end{align}
The sum $\langle i,j\rangle$ counts each hopping once.
For the normalized occupation basis $\ket{x}$, with $x_i\in\mathbb Z_{\ge0}$,
\begin{align*}
    \hat a_i^\dagger\ket{x}&=\sqrt{x_i+1}\ket{x+e_i},\\
    \hat a_i\ket{x}&=\sqrt{x_i}\ket{x-e_i},
\end{align*}
where $e_i$ is the unit occupation vector at site $i$, and the second expression is zero for $x_i=0$.
Setting
\begin{equation}
    \chi(t)=\int_0^t u(s)\,ds,\qquad
    V(t)=e^{i\chi(t)A_{\mathrm B}},
\end{equation}
gives exactly
\begin{equation}
    \widetilde H_{\mathrm B}(t)
    =e^{i\chi(t)A_{\mathrm B}}B_{\mathrm B}e^{-i\chi(t)A_{\mathrm B}},
\end{equation}
with the laboratory-frame endpoint factor in \Cref{eqn:squash_frame_endpoint} and $A=A_{\mathrm B}$.

We realize the sampled conjugation parameters using short interaction pulses separated by intervals with $u(t)=0$. Each pulse changes the accumulated phase $\chi(t)=\int_0^t u(s)\,ds$ to the next sampled value, while each subsequent evolution interval keeps this phase constant for the sampled evolution time. We assume hopping during the interaction pulses is negligible. In the ideal instantaneous-pulse limit, the interaction-picture propagator is therefore exactly a product of the conjugated evolutions studied above; finite pulse durations introduce an approximation.
Physically, this protocol can be implemented by rapidly increasing the interaction-to-hopping ratio $u(t)/J$ by tuning $u(t)$ while keeping the hopping fixed, and returning $u(t)$ to zero during the evolution intervals, provided the pulses are short enough that hopping during them is negligible.

We consider a finite lattice of $N$ sites and restrict to the sector $\mathcal H_M$ with a fixed total number $M$ of bosons. This sector has dimension
\[
    d=\binom{M+N-1}{M}.
\]
Both $A_{\mathrm B}$ and $B_{\mathrm B}$ preserve $\mathcal H_M$, so their restriction to this sector is exact. Our main results apply, provided the connected Lie subgroup generated on this sector is closed.

The main idea is that conjugation gives hopping an occupation-dependent phase.
Let $i\ne j$, $x_j\ge1$, and $y=x+e_i-e_j$. The phase comes from the difference between the interaction eigenvalues of $\ket{y}$ and $\ket{x}$.

First, the hopping operator acts as
\[
\hat a_i^\dagger\hat a_j\ket{x}
=\sqrt{x_j}\sqrt{x_i+1}\ket{y},
\]
so
\[
\bra{y}B_{\mathrm B}\ket{x}
=-J_{ij}\sqrt{(x_i+1)x_j}.
\]

Next, write
\[
A_{\mathrm B}\ket{x}=E_x\ket{x},
\qquad
E_x=\frac12\sum_k x_k(x_k-1).
\]
Since $y_i=x_i+1$, $y_j=x_j-1$, and all other occupations are unchanged,
\[
\begin{aligned}
E_y-E_x
&=\frac12\Big[
(x_i+1)x_i-x_i(x_i-1)\\
&\qquad +(x_j-1)(x_j-2)-x_j(x_j-1)
\Big]\\
&=\frac12\big[2x_i-2(x_j-1)\big]\\
&=x_i-x_j+1.
\end{aligned}
\]

Finally, using $\widetilde H_{\mathrm B}(t)=e^{i\chi(t)A_{\mathrm B}}B_{\mathrm B}e^{-i\chi(t)A_{\mathrm B}}$,
\[
\begin{aligned}
\bra{y}\widetilde H_{\mathrm B}(t)\ket{x}
&=e^{i\chi(t)E_y}
  \bra{y}B_{\mathrm B}\ket{x}
  e^{-i\chi(t)E_x}\\
&=e^{i\chi(t)(E_y-E_x)}
  \bra{y}B_{\mathrm B}\ket{x}\\
&=-J_{ij}\sqrt{(x_i+1)x_j}\,
  e^{i\chi(t)(x_i-x_j+1)}.
\end{aligned}
\]
If $x_j=0$, then $\hat a_j\ket{x}=0$, so this hop has zero amplitude.

\paragraph*{Fermi--Hubbard: interaction control.}
For spin-$\frac12$ fermions on a finite lattice with spin-independent hopping, take~\cite{arovas2022}
\begin{align}
    H_{\mathrm F}(t)&=B_{\mathrm F}+u(t)A_{\mathrm F},\\
    B_{\mathrm F}&=-\sum_{\langle i,j\rangle,\sigma}
    \left(J_{ij}\hat c_{i\sigma}^\dagger\hat c_{j\sigma}
    +J_{ij}^*\hat c_{j\sigma}^\dagger\hat c_{i\sigma}\right),\\
    A_{\mathrm F}&=\sum_i\hat n_{i\uparrow}\hat n_{i\downarrow},
    \qquad \hat n_{i\sigma}=\hat c_{i\sigma}^\dagger\hat c_{i\sigma}.
\end{align}
Since $\hat n_{i\sigma}^2=\hat n_{i\sigma}$, the total on-site occupation $\hat n_i=\hat n_{i\uparrow}+\hat n_{i\downarrow}$ satisfies
\begin{equation}
    A_{\mathrm F}=\frac12\sum_i(\hat n_i^2-\hat n_i).
\end{equation}
As in the bosonic case, independently sampled conjugation parameters are implemented using interaction pulses followed by zero-interaction evolution intervals, with the product description exact in the instantaneous-pulse limit. Taking $V(t)=e^{i\chi(t)A_{\mathrm F}}$ and $\chi(t)=\int_0^t u(s)\,ds$ yields
\begin{equation}
    \widetilde H_{\mathrm F}(t)
    =e^{i\chi(t)A_{\mathrm F}}B_{\mathrm F}e^{-i\chi(t)A_{\mathrm F}}.
\end{equation}

\paragraph*{Superconducting circuits: microwave phase control.}
Consider capacitively coupled transmons driven at a common microwave frequency.
In a fixed frame rotating at that frequency, the low-energy oscillator model, after the rotating-wave approximation, takes the form~\cite{krantz2019}
\begin{align}
    H_{\mathrm{sc}}(\phi)
    &=D_{\mathrm{sc}}+\frac12\sum_i\Omega_i
    \left(e^{i\phi}\hat b_i^\dagger+e^{-i\phi}\hat b_i\right),\\
    D_{\mathrm{sc}}
    &=\sum_i\left[\Delta_i\hat n_i
    +\frac{\alpha_i}{2}\hat n_i(\hat n_i-1)\right]
    \nonumber\\
    &\quad+\sum_{i<j}\left(g_{ij}\hat b_i^\dagger\hat b_j
    +g_{ij}^*\hat b_j^\dagger\hat b_i\right).
\end{align}
Here $\hat b_i$ lowers the transmon excitation, $\hat n_i=\hat b_i^\dagger\hat b_i$, $\Delta_i$ is the detuning from the microwave frequency, $\alpha_i$ is the anharmonicity, and $g_{ij}$ is the exchange coupling.
The drive amplitudes $\Omega_i$, detunings $\Delta_i$, anharmonicities $\alpha_i$, and exchange couplings $g_{ij}$ are held fixed; we only vary the common microwave drive phase $\phi$.
We keep only the lowest $d_i$ levels of each transmon, under the assumption in the experiment that population of higher levels remains negligible during the evolution. To have a qubit, we keep $d_i = 2$.

Because $H_{\mathrm{sc}}$ is a driven system with tunable phase, we can apply the same strategy as for the Rydberg example.
Let $A_{\mathrm{sc}}=\sum_i\hat n_i$ and $B_{\mathrm{sc}}=H_{\mathrm{sc}}(0)$.
All terms in $D_{\mathrm{sc}}$ conserve total excitation number, while
\begin{equation}
    e^{i\phi A_{\mathrm{sc}}}\hat b_i^\dagger
    e^{-i\phi A_{\mathrm{sc}}}=e^{i\phi}\hat b_i^\dagger.
\end{equation}
Therefore,
\begin{equation}\label{eqn:sc_phase_covariance}
    H_{\mathrm{sc}}(\phi)
    =e^{i\phi A_{\mathrm{sc}}}B_{\mathrm{sc}}e^{-i\phi A_{\mathrm{sc}}}.
\end{equation}
Thus microwave phase control realizes the $H(\phi)$ family directly within this effective model, including its anharmonic and exchange terms.
Projecting each transmon onto its lowest two levels gives the driven-qubit model.
We use a frame rotating at the fixed microwave reference frequency. The $\phi$-dependent conjugation is not an interaction-picture transformation. It relates Hamiltonians produced by different drive phases, all expressed in the same fixed overall rotating frame.

\paragraph*{Trapped ions: interaction-axis phase control.}
Bichromatic laser fields, containing two frequency components,
can generate effective Ising interactions between trapped-ion
qubits through their shared vibrations~\cite{monroe2021}.
The excitation quanta of these vibrations are called phonons.
Besides the main resonance, which changes the qubit state
without changing the vibration, there are additional
resonances called motional sidebands. At the first blue
sideband, exciting the qubit also creates one phonon;
at the first red sideband, exciting the qubit removes one
phonon. For qubit frequency $\omega_0$ and vibrational
frequency $\omega_m$, these resonances occur at
$\omega_0+\omega_m$ and $\omega_0-\omega_m$, respectively.

The two laser components are placed near the red and blue
sidebands, with detunings large compared with the relevant
spin-motion coupling strengths. The resulting weak coupling
to states with different phonon numbers mediates an interaction
between the spins. Neglecting the small laser-induced changes
in motion and residual spin-motion entanglement gives an
effective Hamiltonian acting only on the spins~\cite[Eqn.~20]{monroe2021}:
\begin{equation}\label{eqn:trapped_ion}
    H_{\mathrm{ion}}(\phi)
    =\sum_{i<j}J_{ij}\sigma_i(\phi)\sigma_j(\phi)
    +\sum_i h_i\sigma_i^z,
\end{equation}
where
\begin{equation*}
    \sigma_i(\phi)=\cos\phi\,\sigma_i^x+\sin\phi\,\sigma_i^y.
\end{equation*}
Here $\phi$ is the common spin phase set by the bichromatic fields, and $J_{ij}$ and the effective fields $h_i$ are fixed.

We vary the spin phase while holding the motional phase and the parameters determining $J_{ij}$ fixed.
We assume that the motional effects remain small during the entire course of evolution so \Cref{eqn:trapped_ion} is a valid model.

Define
\begin{equation}
    A_{\mathrm{ion}}=-\frac12\sum_i\sigma_i^z,\quad
    B_{\mathrm{ion}}=\sum_{i<j}J_{ij}\sigma_i^x\sigma_j^x
    +\sum_i h_i\sigma_i^z.
\end{equation}
Since 
\begin{align*}
e^{i\phi A_{\mathrm{ion}}}&\sigma_i^x e^{-i\phi A_{\mathrm{ion}}}\\
&=
\begin{pmatrix}
e^{-i\phi/2} & 0\\
0 & e^{i\phi/2}
\end{pmatrix}
\begin{pmatrix}
0 & 1\\
1 & 0
\end{pmatrix}
\begin{pmatrix}
e^{i\phi/2} & 0\\
0 & e^{-i\phi/2}
\end{pmatrix}\\
&=
\begin{pmatrix}
0 & e^{-i\phi}\\
e^{i\phi} & 0
\end{pmatrix}\\
&=\cos\phi\,\sigma_i^x+\sin\phi\,\sigma_i^y\\
&=\sigma_i(\phi),
\end{align*}
we have
\begin{equation}\label{eqn:ion_phase_covariance}
    H_{\mathrm{ion}}(\phi)
    =e^{i\phi A_{\mathrm{ion}}}B_{\mathrm{ion}}e^{-i\phi A_{\mathrm{ion}}},
\end{equation}
and
\begin{align*}
    \sigma_i(\phi)\sigma_j(\phi)
    ={}&\cos^2\phi\,\sigma_i^x\sigma_j^x
    +\sin^2\phi\,\sigma_i^y\sigma_j^y
    \nonumber\\
    &+\sin\phi\cos\phi
    (\sigma_i^x\sigma_j^y+\sigma_i^y\sigma_j^x).
\end{align*}

An additional resonant global carrier drive can also be included:
\[
C_{\mathrm{ion}}(\varphi)
=\frac12\sum_i\Omega_i\sigma_i(\varphi).
\]

To retain a single control phase, set
\[
\varphi=\phi+\varphi_0,
\]
where $\varphi_0$ is a fixed relative phase.
Then
\[
H_{\mathrm{tot}}(\phi,\phi+\varphi_0)
=e^{i\phi A_{\mathrm{ion}}}
\bigl[B_{\mathrm{ion}}+C_{\mathrm{ion}}(\varphi_0)\bigr]
e^{-i\phi A_{\mathrm{ion}}}.\\
\]

\paragraph*{Summary of the Hamiltonians.}
The table lists the reference Hamiltonian $B$ and conjugating generator $A$ for each model.
The Hubbard entries describe the interaction-picture construction, and the trapped-ion entry includes the optional carrier with $\varphi_0=0$.

\begin{widetext}
\begin{tabular}{l c c}
\hline
System & $B$ & $A$ \\
\hline
Rydberg array
&
$\displaystyle
\begin{aligned}
&\frac{\Omega}{2}\sum_i\sigma_i^x
-\sum_i\Delta_i\hat n_i\\
&\quad+\sum_{i<j}J_{ij}\hat n_i\hat n_j
\end{aligned}$
&
$\displaystyle\sum_i\hat n_i$
\\[1em]
\hline
Bose--Hubbard
&
$\displaystyle
-\sum_{\langle i,j\rangle}
\left(J_{ij}\hat a_i^\dagger\hat a_j
+\mathrm{h.c.}\right)$
&
$\displaystyle
\frac12\sum_i\hat n_i(\hat n_i-1)$
\\[1em]
\hline
Fermi--Hubbard
&
$\displaystyle
-\sum_{\langle i,j\rangle,\sigma}
\left(J_{ij}\hat c_{i\sigma}^\dagger\hat c_{j\sigma}
+\mathrm{h.c.}\right)$
&
$\displaystyle
\sum_i\hat n_{i\uparrow}\hat n_{i\downarrow}$
\\[1em]
\hline
Superconducting array
&
$\displaystyle
\begin{aligned}
&\sum_i\left[
\Delta_i\hat n_i
+\frac{\alpha_i}{2}\hat n_i(\hat n_i-1)
\right]\\
&+\sum_{i<j}
\left(g_{ij}\hat b_i^\dagger\hat b_j
+\mathrm{h.c.}\right)\\
&+\frac12\sum_i\Omega_i
(\hat b_i^\dagger+\hat b_i)
\end{aligned}$
&
$\displaystyle\sum_i\hat n_i$
\\[1em]
\hline
Trapped ions
&
$\displaystyle
\begin{aligned}
&\sum_{i<j}J_{ij}\sigma_i^x\sigma_j^x
+\sum_i h_i\sigma_i^z\\
&\quad+\frac12\sum_i\Omega_i\sigma_i^x
\end{aligned}$
&
$\displaystyle-\frac12\sum_i\sigma_i^z$
\\[1em]

\hline
\end{tabular}
\end{widetext}

\end{document}